\documentclass[11pt,letterpaper,UKenglish,cleveref, autoref]{article}
\usepackage[utf8]{inputenc}

\usepackage{thmtools}
\usepackage{thm-restate}

\usepackage{fullpage}
\usepackage[colorlinks=true,linkcolor=blue,citecolor=blue,urlcolor=blue]{hyperref}

\usepackage{amsmath,amsfonts,amssymb,amsthm}
\usepackage{latexsym}
\usepackage{xcolor}
\usepackage{caption}
\usepackage{subcaption}

\usepackage{bm}
\usepackage{graphicx}

\usepackage{blkarray}
\usepackage{enumitem}   
\usepackage{float}

\newcommand{\squeezelist}{\setlength{\itemsep}{0pt}}

\usepackage{hyperref}

\newtheorem{theorem}{Theorem}%
\newtheorem{lemma}[theorem]{Lemma}
\newtheorem{proposition}[theorem]{Proposition}
\newtheorem{corollary}[theorem]{Corollary}
\newtheorem{claim}[theorem]{Claim}
\newtheorem{observation}[theorem]{Observation}

\DeclareMathOperator*{\argmin}{argmin}

\newcommand{\newcomment}[1]{{\color{black}#1}}
\newcommand{\revised}[1]{{\color{black}#1}}

\newcommand{\changed}[1]{{\color{black}#1}}
\newcommand{\fchanged}[1]{{\color{black}#1}}

\newcommand{\newestchanged}[1]{{#1}}
\newcommand{\attention}[1]{{#1}}

\def\defn#1{\textit{\textbf{\boldmath #1}}}

\title{The Visibility Center of a Simple Polygon}

\author{Anna Lubiw and Anurag Murty Naredla}

\begin{document}

\maketitle

\begin{abstract}
We introduce 
the \emph{visibility center} of a set of points inside a %
polygon---a point 
$c_V$ such that the maximum geodesic distance from $c_V$ to see any point in the set is minimized.
For a simple polygon of $n$ vertices and a set of $m$ points inside it, we give an $O((n+m) \log {(n+m)})$ time algorithm to find the visibility center.  We find the  visibility center of \emph{all} points in a simple polygon in $O(n \log n)$ time.

Our algorithm reduces the visibility center problem to the problem of finding the geodesic center of a set of half-polygons inside a polygon, which is of independent interest.  We give an $O((n+k) \log (n+k))$ time algorithm for this problem, where $k$ is the number of half-polygons.

\end{abstract}

\pagenumbering{arabic} 
\section{Introduction}

Suppose you want to guard a polygon and you have many sensors but only one guard to check on the sensors.
The guard must be positioned at a point $c_V$ in the polygon such that when a sensor at any query point $u$ sends an alarm, the guard travels from $c_V$ on a shortest path inside the polygon to \emph{see} point $u$; 
 the goal is to minimize the maximum distance the guard must travel.
More precisely, we must choose $c_V$ to minimize the maximum, over points $u$, of the geodesic distance from $c_V$ to a point that sees $u$.
The optimum guard position $c_V$ is called the \defn{visibility center} of the set $U$ of possible query points.
See Figure~\ref{fig:figure_with_visibility_center_paths}.
We give an $O((n+m) \log {(n+m)})$ time algorithm to find the visibility center of a set $U$ of size $m$ in an $n$-vertex simple polygon.
To find the visibility center of \textit{all} points inside a simple polygon, 
we can restrict our attention to the 
vertices of the polygon, which 
yields an $O(n \log n)$ time algorithm.

\begin{figure}[h]
\begin{subfigure}{0.5\textwidth}
  \centering
\includegraphics[width=\linewidth]{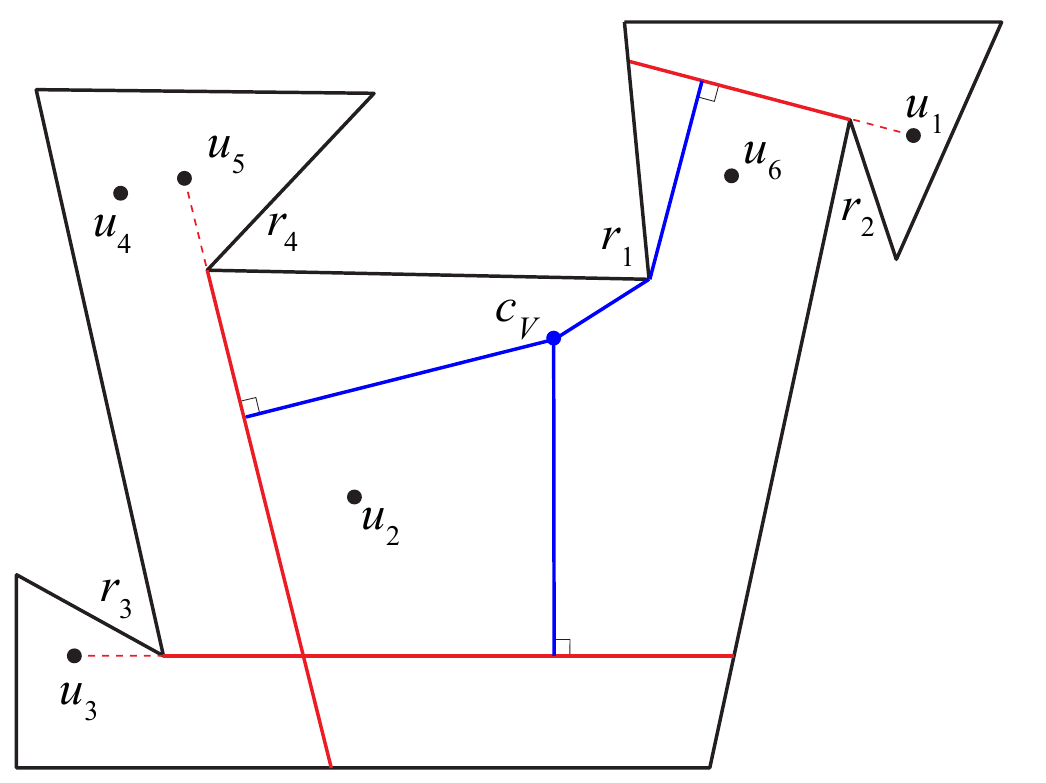}
\end{subfigure}
\begin{subfigure}{0.5\textwidth}
  \centering
\includegraphics[width=\linewidth]{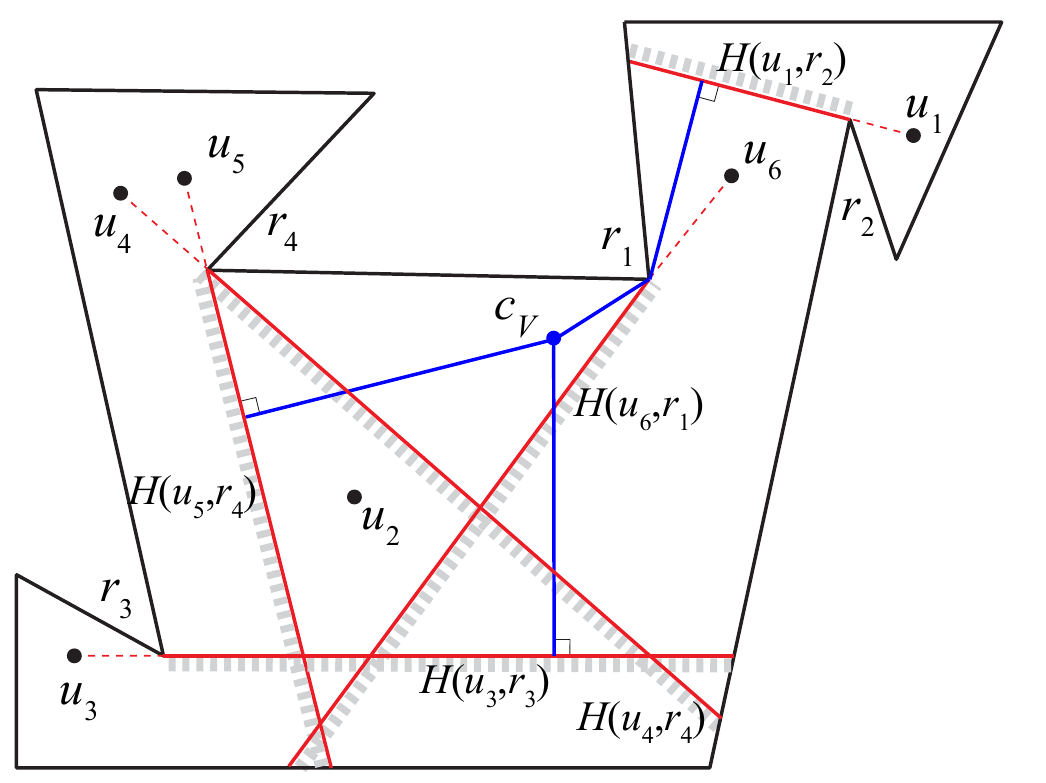}
\end{subfigure}
\caption{
(left) Point $c_V$ is the \defn{visibility  center} of points $U = \{u_1, \ldots, u_6\}$.
Starting from $c_V$, the three points we need to travel (equally) farthest to see are  $u_1, u_3$ and $u_5$. The shortest paths (in blue) to 
see these points must reach the  half-polygons bounded by the chords (in red) emanating from the points.
(right) 
Equivalently, $c_V$ is the geodesic center of five half-polygons (each shown as a red boundary chord shaded on one  side).}
\label{fig:figure_with_visibility_center_paths}
\end{figure}

To the best of our knowledge, the idea of visibility centers is new, 
though it is a very natural concept
that combines %
two significant branches of computational geometry:
visibility problems~\cite{ghosh2007visibility};
and center problems and farthest Voronoi diagrams~\cite{aurenhammer2013voronoi}. 

There is a long history of finding ``center points'', for various definitions of ``center''. 
The most famous of these is Megiddo's linear time algorithm~\cite{megiddo_linear} to find the 
center of a set of points in the plane (Sylvester's
``smallest circle'' problem). 

Inside a polygon the relevant distance measure is not the Euclidean distance but rather the shortest path, or  
\emph{geodesic}, distance. 
The \defn{geodesic center} of a simple polygon is a point $p$ that minimizes the maximum geodesic distance from $p$ to any point $q$ of the polygon, or equivalently, the maximum geodesic distance from $p$ to any vertex of the polygon. 
Pollack, Sharir, and Rote~\cite{pollack_sharir} gave an $O(n \log n)$ time divide-and-conquer algorithm to find  the geodesic center of a polygon. 
Our algorithm builds on theirs. 
A more recent algorithm by Ahn et al.~\cite{linear_time_geodesic} finds the geodesic center of a polygon in linear time.
Another notion of the center of a polygon is the link center, which can  be found in $O(n \log n)$ time~\cite{djidjev1992ano}.

Center problems are closely related to farthest Voronoi diagrams, since the center is %
a vertex of the corresponding farthest Voronoi diagram, 
\changed{or a point on an edge of the Voronoi diagram in case the center has only two farthest sites}.  Finding the farthest Voronoi diagram of points in the plane takes $\Theta(n \log n)$ time---thus is it strictly harder to find the farthest Voronoi diagram than to find the center.
However, working inside a simple polygon helps for farthest Voronoi diagrams: 
the farthest geodesic Voronoi diagram of the vertices of a polygon can be found in time $O(n \log \log n)$~\cite{oh2020geodesic}. 
\newcomment{Generalizing the two scenarios (points in the plane, and polygon vertices), yields the} 
problem of finding the farthest Voronoi diagram of $m$ points in a polygon, which was first solved by Aronov et al.~\cite{aronov1993furthest} with run-time $O((n+m)\log (n+m))$, and improved in a sequence of papers~\cite{oh2020geodesic,barba2019optimal,oh2020voronoi}, with the current best run-time of 
$O(n + m \log m)$~\cite{wang2021}.

Turning to visibility problems in a polygon, 
there are algorithms for the  %
``quickest visibility problem''---to find the shortest path from point $s$ to see point $q$, and to solve the query version where $s$ is fixed and $q$ is a query point~\cite{arkin2016shortest,wang2019quickest}.  For a simple polygon~\cite{arkin2016shortest}, the preprocessing time and space are $O(n)$ and 
the query time is $O(\log n)$.
\changed{We do not use these results}
in our algorithm to find the visibility center $c_V$, but they are useful afterwards to find the actual shortest path from $c_V$ to see a query point. 

A more basic version of our problem is to find, if there is one, a point that sees all points in $U$.  The set of such points is  the \defn{kernel} of $U$.  
When $U$ is the set of vertices, the kernel can be found in linear time~\cite{lee1979optimal}. For a general set $U$, Ke and O'Rourke~\cite{ke1989computing} gave an $O(n + m\log (n + m))$ time algorithm to compute the kernel, and we use some of their results in our algorithm.  

Another problem somewhat similar to the visibility center problem is 
the watchman problem~\cite{chin1991shortest,touring}---to find a minimum length tour from which a single guard can  see the whole polygon.  Our first  step is similar in flavour to the first step for the watchman problem, namely, to replace the condition of ``seeing'' everything by a condition of visiting certain ``essential chords''.

\paragraph*{Our Results}
The \defn{distance to visibility} from a point $x$ to point $u$ in $P$, denoted $d_V(x,u)$ is the minimum distance in $P$ from $x$ to a point $y$ such that $y$ sees $u$.
For a set of points $U$ in $P$, 
the \defn{visibility radius} of $x$ with respect to $U$ is $r_V(x,U) := \max \{d_V(x,u) : u \in U\}$.
The \defn{visibility center} $c_V$ of $U$ is a point $x$ that minimizes $r_V(x,U)$.  
Our main result is:

\begin{theorem}
\label{thm:vis-center}
There is an algorithm to find the visibility center of a point set $U$ of size $m$ in a simple $n$-vertex polygon $P$ with run-time $O((n+m)\log (n+m))$. 
\end{theorem}

The key to our algorithm is to reformulate the visibility center problem in terms of distances to certain \emph{half-polygons} inside the polygon. %
We illustrate the idea by means of the example in Figure~\ref{fig:figure_with_visibility_center_paths} where the visibility center of the 6-element point set $U$ is the 
 \emph{geodesic center} of a set of five half-polygons.

More generally, we will reduce the problem of finding the visibility center to the problem of finding a geodesic center of a linear number of half-polygons.
The input to this problem is a set $\cal H$ of $k$ half-polygons (see Section~\ref{sec:preliminaries} for precise definitions) and the goal is to find a \emph{geodesic center} $c$ that minimizes the maximum distance from $c$ to a half-polygon.  
More precisely, 
the \defn{geodesic radius} from a point $x$ to $\cal H$ is $r(x,{\cal H}) := \max \{d(x,H) : H \in {\cal H}\}$, and 
the \defn{geodesic center} $c$ of $\cal H$ is a point $x$ that minimizes $r(x,{\cal H})$.
Our second main result is:

\begin{theorem}
\label{thm:geo-center}
There is an algorithm to find the geodesic center of a set $\cal H$ of $k$ half-polygons in a simple $n$-vertex polgyon $P$ with run-time $O((n+k)\log (n+k))$.
\end{theorem}

Our algorithm extends the divide-and-conquer approach 
that Pollack et al.~\cite{pollack_sharir} used to compute the geodesic center of the vertices of a simple polygon.

Our main motivation for finding the geodesic center of half-polygons is to find the visibility center, but the geodesic center of half-polygons is of independent interest.
Euclidean problems of a similar flavour are to find the center (or the farthest Voronoi diagram) of line segments or convex polygons in the plane~\cite{bhattacharya1994optimal, jadhav1996optimal}. These problems are less well-studied than the case of point sites (e.g., see~\cite{aurenhammer2006farthest} for remarks on this). \newcomment{The literature for geodesic centers is even more sparse, focusing almost exclusively on geodesic centers of points in a polygon.}
It is thus
interesting that the center of half-polygons inside a polygon can be found efficiently.
As a special case, we can find the geodesic center of the edges of a simple polygon 
in $O(n \log n)$ time.

The reduction from the visibility center problem to the geodesic center of half-polygons is in  
Section~\ref{section:essential-half-polygons}. The run time is $O((n+m) \log (n+m))$.
The algorithm that proves Theorem~\ref{thm:geo-center} \changed{(finding the geodesic center of half-polygons)} is in Section~\ref{section:half-polygon-center}.  Together these prove Theorem~\ref{thm:vis-center} \changed{(finding the visibility center)}.

\section{Preliminaries}
\label{sec:preliminaries}

We add a few more basic definitions to augment the main definitions given above.
We work with a simple polygon $P$ of $n$ vertices whose boundary $\partial P$ is directed clockwise. A \defn{chord} of $P$ is a line segment inside $P$ 
\changed{with endpoints on $\partial P$.} 
Any chord divides $P$ into two \changed{weakly simple} \defn{half-polygons}.
A half-polygon is specified by its \changed{defining} chord $(p,q)$ with the convention that the half-polygon contains the path clockwise from $p$ to $q$. 

The \defn{geodesic distance $d(x,y)$} (or simply, \defn{distance}) between two points $x$ and $y$ in $P$ is the length of the \defn{shortest path $\pi(x,y)$} in $P$ from $x$ to $y$. For half-polygon $H$, 
the \defn{geodesic distance $d(x,H)$} is the minimum distance from $x$ to a point in $H$.

Points $x$ and $y$ in $P$ are \defn{visible} ($x$ ``sees'' $y$) if the segment $xy$ lies inside $P$.  The \defn{distance to visibility} from $x$ to $u$, denoted $d_V(x,u)$ is the minimum distance from $x$ to a point $y$ such that $y$ sees $u$. If $x$ sees $u$, then this distance is 0, and otherwise it is the distance from $x$ to the half-polygon defined as follows.  
Let $r$ be the last reflex vertex on the shortest path from $x$ to $u$.  Extend the ray $\overrightarrow{ur}$ from $r$ until it hits the polygon boundary $\partial P$ at a point $p$ to obtain a chord $rp$ 
(which is an edge of the \defn{visibility polygon} of $u$).
Of the two half-polygons defined by $rp$, 
let $H(u,r)$ be the one that contains $u$.  See Figure~\ref{fig:figure_with_visibility_center_paths}.

\begin{observation}
\label{obs:vis-to-half-polygon}
$d_V(x,u) = d(x,H(u,r))$.
\end{observation}

\changed{In the remainder of this section we establish the basic result that 
the visibility center of a set of points $U$ and the geodesic center of a set of half-polygons $\cal H$ are unique except in very special cases, and that two or three tight constraints suffice to determine the centers.} 
We explain this for the geodesic center of half-polygons, but the same argument works for the visibility center (or, alternatively, one can use the  reduction from the visibility center to the geodesic center in Section~\ref{section:essential-half-polygons}).
\changed{Note that 
if the geodesic radius is 0, then any point in the intersection of the half-polygons is a geodesic center.}  

\begin{claim}
\label{claim:2-or-3-vectors}
\changed{Suppose that the geodesic radius $r$ satisfies $r > 0$.}
There is a set ${\cal H}' \subseteq {\cal H}$ of two or three half-polygons such that the set of  geodesic centers of ${\cal H}$ is equal to the set of geodesic centers of ${\cal H}'$ and furthermore 
\begin{enumerate}
\item if ${\cal H}'$ has size 3 then the geodesic center is unique 
(e.g., see Figure~\ref{fig:figure_with_visibility_center_paths})

\item if ${\cal H}'$ has size 2 then either the geodesic center is unique %
or  
the two half-polygons of ${\cal H}'$ have chords that are parallel and the geodesic center consists of a  line segment parallel to them and midway between them.
\end{enumerate}
\end{claim}

\changed{The proof of this claim depends on a basic convexity property of the geodesic distance function that was proved for the case of distance to a vertex by  
Pollack et al.~\cite[Lemma 1]{pollack_sharir} and that we extend to a half-polygon.
}

\medskip

\noindent

\medskip

\begin{figure}[H]
\centering
\includegraphics[width=0.7\textwidth]{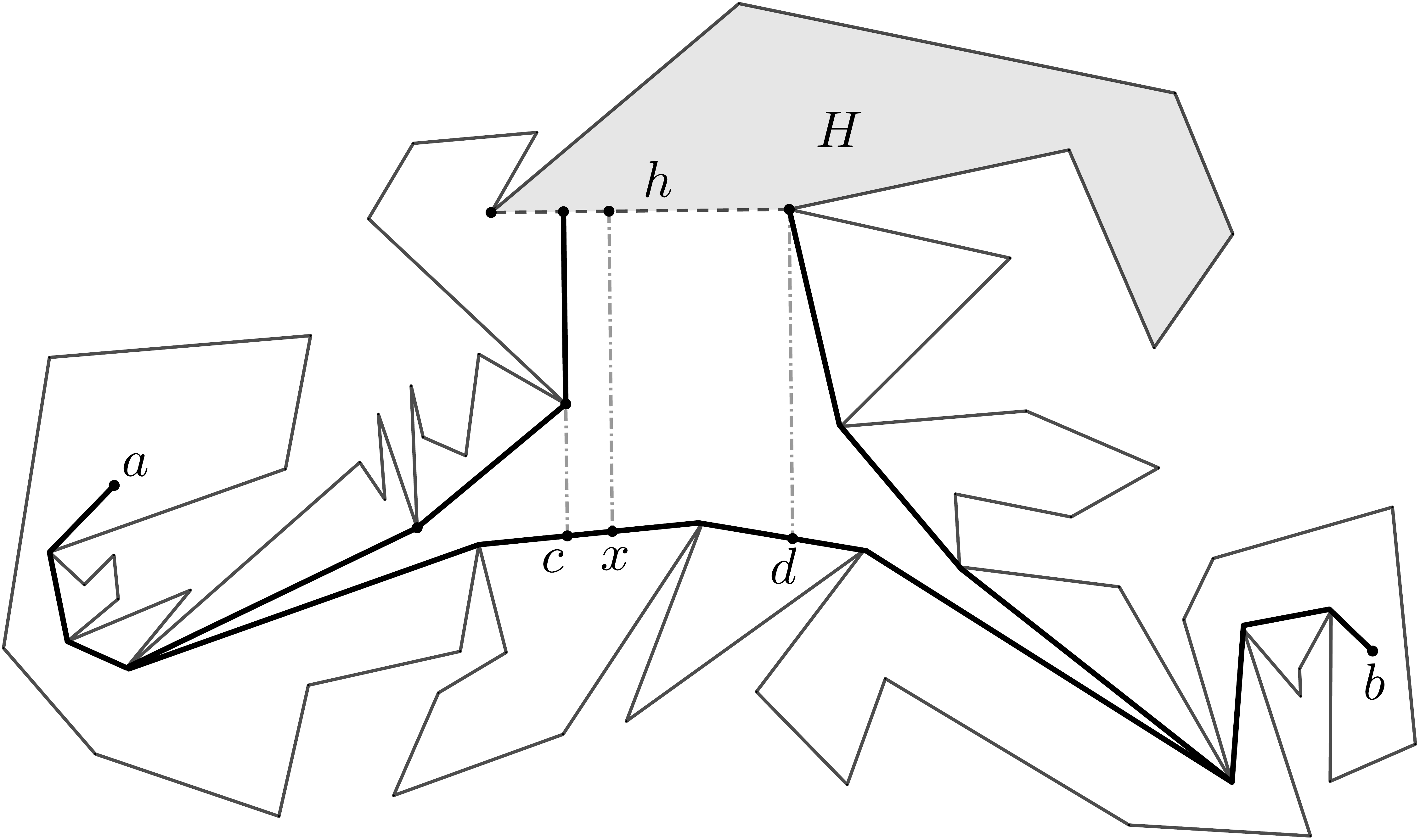}
\caption{
Proving that 
\revised{$d(x,H)$} is a convex function of $d(a,x)$. 
}
\label{fig:convexity-proof}
\end{figure}

\newestchanged{A subset $Q$ of  $P$ is \defn{geodesically convex} if for any two points $a$ and $b$ in $Q$, the shortest (or ``geodesic'') path $\pi(a,b)$ in $P$ is contained in $Q$.
A function $f$ defined on $P$ is \defn{geodesically convex} if $f$ is convex on every geodesic path $\pi(a,b)$ in $P$, i.e., for points $x \in \pi(a,b)$, $f(x)$ is a convex function of $d(a,x)$.
}

\begin{lemma}
\label{lemma:distance_to_half_polygon_is_convex}
\newestchanged{For any half-polygon $H$, the distance function $d(x,H)$ is geodesically convex.
Furthermore, on any geodesic path $\pi(a,b)$ with $a$ and $b$ outside $H$, the minimum of $d(x,H)$ occurs at a point or along a line segment parallel to $h$, the defining chord of $H$.}
\end{lemma}
\begin{proof}
Pollack et al.~\cite{pollack_sharir} proved the version of this where $H$ is replaced by a point\changed{---in particular, they proved that the distance function is strictly convex which implies that the minimum occurs at a point}.

If $a$ and $b$ are inside $H$, then so is $\pi(a,b)$ and the distance function is constantly 0.  So suppose $a \notin H$.  If the path $\pi(a,b)$ intersects $H$, then it does so at only one point $b^* \in h$ (otherwise $h$ provides a shortcut for part of the path).  
\changed{Then $\pi(b^*,b)$ lies inside $H$ and it suffices to prove convexity 
for the path $\pi(a,b^*)$.}  In other words, we may assume that $b$ is not in the interior of $H$.

The shortest paths $\pi(x,H)$ for $x \in \pi(a,b)$ reach a subinterval $[h_1,h_2]$ of $h$. 
See Figure~\ref{fig:convexity-proof}.  In case this subinterval is a single point, i.e., $h_1=h_2$, the convexity result of Pollack et al.~proves the lemma.
Otherwise, since shortest paths do not cross,  there are points $c$ and $d$ on $\pi(a,b)$ such that for $x \in \pi(a,c)$ the path $\pi(x,H)$ arrives at $h_1$,  for $x \in \pi(c,d)$ the path $\pi(x,H)$ is a straight line segment reaching $h$ at a right angle, and for $x \in \pi(d,b)$ the path $\pi(x,H)$ arrives at $h_2$.  
The convexity result of Pollack et al.~applies to the paths arriving at the points $h_1$ and $h_2$.  

It remains to consider $x \in \pi(c,d)$. 
As $x$ moves along $\pi(c,d)$, the endpoint of $\pi(x,H)$ moves continuously with a one-to-one mapping along the segment $[h_1,h_2]$.  Since the curve $\pi(c,d)$ is convex, this implies that  $d(x,H)$ is a convex function of $d(a,x)$ for $x \in \pi(c,d)$. 
\changed{Furthermore, the minimum is unique unless a segment of the geodesic is parallel to $h$.
}

Finally, one can verify that convexity holds at the points $c$ and $d$, i.e., that the three convex functions join to form a single convex function. 
\end{proof}

\newestchanged{
Because the intersection of geodesically convex sets is geodesically convex, and the max of geodesically convex functions is geodesically convex, we get the following consequences.

\begin{corollary}
\label{cor:geodesic-convexity}
The geodesic radius function $r(x,{\cal H})$ is geodesically convex.
The geodesic ball $B(t,H) := \{x \in P : d(x,H) \le t \}$ for any half-polygon $H$, and the geodesic ball $B(t) :=\{x \in P : r(x,{\cal H}) \le t \}$  are geodesically convex. 
\end{corollary}
}

\begin{proof}[Proof of Claim~\ref{claim:2-or-3-vectors}]
\changed{
The set of geodesic centers is $C := \{x \in P : d(x,H) \le r \text{ for all } H \in {\cal H} \}$.
\newestchanged{By Corollary~\ref{cor:geodesic-convexity}, $C$ is geodesically convex.}
If $C$ contains two distinct points $x_1$ and $x_2$ then it contains the geodesic path $\gamma = \pi(x_1,x_2)$.  
By Lemma~\ref{lemma:distance_to_half_polygon_is_convex}, along $\gamma$, for each $H \in {\cal H}$, the minimum of the distance function $d(x,H)$ occurs at a point or along a line segment parallel to $h$.  This implies that $\gamma$ can only be a single line segment parallel to two of the half-polygons of $\cal H$, which is Case 2 of the Claim.

For Case 1, let us now suppose that 
$C$ consists of a single point.
Because the boundary of each
geodesic ball $\{x \in P : d(x,H) \le r \}$ consists of circular arcs and line segments, the single point $C$
is uniquely determined as the intersection of some circular arcs and line segments, and three of those suffice to determine the point.
}
\end{proof}

\section{Reducing the Visibility Center to the Center of Half-Polygons}
\label{section:essential-half-polygons}
In this section we reduce the problem of finding the visibility center of a set of points $U$ in a polygon $P$ to the problem of finding the geodesic center of a linear number of ``essential'' half-polygons  $\cal H$, which is solved in Section~\ref{section:half-polygon-center}.

By Observation~\ref{obs:vis-to-half-polygon} (and see Figure~\ref{fig:figure_with_visibility_center_paths}) 
the visibility center of $U$ 
is the geodesic center of 
the set of $O(mn)$ half-polygons $H(u,r)$ where $u \in U$, $r$ is a reflex vertex of $P$ that sees $u$, and $H(u,r)$ is the half-polygon containing $u$ and bounded by the chord that extends
$\overrightarrow{ur}$ from $r$ until it hits  $\partial P$ at a point $t$.
Note that finding $t$ is a ray shooting problem and costs $O(\log n)$ time after an $O(n)$ time preprocessing step~\cite{hershberger1995pedestrian}. 

However,  this set of half-polygons is too large.  
We will find a set $\cal H$ of $O(n)$ ``essential'' half-polygons that suffice, i.e., such that 
the visibility center of $U$ is the geodesic center of the half polygons of $\cal H$.
In fact, we give two possible sets of essential half-polygons, ${\cal H}_{\rm reflex}$ and ${\cal H}_{\rm hull}$, where the latter set can be found more efficiently.  \changed{Although the bottleneck is still the algorithm for geodesic center of half-polygons, it seems worthwhile to optimize the reduction.}

We first observe that any half-polygon that contains another one is redundant. 
For example, in Figure~\ref{fig:figure_with_visibility_center_paths},
$H(u_4,r_4)$ is redundant because it contains $H(u_5,r_4)$.
At each reflex vertex $r$ of $P$, there are at most two minimal half-polygons $H(u,r)$. 
Define ${\cal H}_{\rm reflex}$ to be this set of minimal half-polygons. Note that ${\cal H}_{\rm reflex}$ has size $O(n_r)$ where  $n_r$ is the number of reflex vertices of $P$. 

Observe that for the case of finding the visibility center of 
\emph{all} points of $P$, ${\cal H}_{\rm reflex}$ consists of the half-polygons $H(v,r)$ where $(v,r)$ is an edge of $P$, so ${\cal H}_{\rm reflex}$ can be found in time $O(n +  n_r\log n)$.

For a point set $U$, 
the set ${\cal H}_{\rm reflex}$ was also used by 
Ke and O'Rourke~\cite{ke1989computing} in their algorithm
to compute the kernel of point set $U$ in polygon $P$.  (Recall from the Introduction that the kernel of $U$ is the set of points in $P$ that see all points of $U$.)
They gave a sweep line algorithm 
(``Algorithm 2'') to find  
${\cal H}_{\rm reflex}$  in time $O((n+m) \log (n+m))$. 
To summarize:

\begin{proposition}
The geodesic center of ${\cal H}_{\rm reflex}$ is the visibility center of $U$. 
Furthermore, ${\cal H}_{\rm reflex}$ can be found in time $O((n+m) \log (n+m))$. 
\end{proposition}

In the remainder of this section 
we present a second approach using ${\cal H}_{\rm hull}$ that eliminates the $O(n \log n)$ term. 
\newcomment{This does not change the runtime to find the visibility center, but it means that improving the algorithm to find the geodesic center of half-polygons will automatically improve the visibility center algorithm.}
The idea is that ${\cal H}_{\rm reflex}$ is wasteful in that a single point $u \in U$ can give rise to $n_r$ half-polygons. 
Note that we really only need three half-polygons in an essential set, though the trouble is to find them!

\newcomment{We first eliminate the case where the kernel of $U$ is non-empty (i.e., $r_V=0$) by running the $O(n + m \log (n+m))$ time kernel-finding algorithm of Ke and O'Rourke~\cite{ke1989computing}.}
Next we find ${\cal H}_{\rm hull}$ in two steps.  First make a subset ${\cal H}_0$ as follows.  
Construct $R$,  the geodesic convex hull of $U$ in $P$ in time $O(n + m \log(m+n))$~\cite{guibas1989optimal,toussaint1989computing}. 
For each edge $(u,r)$ of $R$ where $u \in U$ and $r$ is a reflex vertex of $P$, put $H(u,r)$ into ${\cal H}_0$.
Note that ${\cal H}_0$ has size $O(\min\{n_r,m\})$ so ray shooting to find the endpoints of the chords $H(u,r)$ takes time $O(n + \min\{n_r,m\} \log n)$.
Unfortunately, as shown in Figure~\ref{fig:first_phase_relCH}, ${\cal H}_0$
can miss an essential half-polygon.

Next, construct a geodesic center $c_0$ of ${\cal H}_0$ using the algorithm of Section~\ref{section:half-polygon-center}.
\newcomment{(Note that the geodesic center can be non-unique and in such cases $c_0$ denotes any one point from the set of geodesic centers.)}
Then repeat the above step for $U \cup \{c_0\}$, %
\newcomment{more precisely,} %
construct $R'$, the geodesic convex hull of  $U \cup \{c_0\}$ in $P$ and 
for each edge $(u,r)$ of $R'$  where $u \in U$ and $r$ is a reflex vertex of $P$, add ${H}(u,r)$ to ${\cal H}_0$. 
This defines ${\cal H}_{\rm hull}$.
Again, ${\cal H}_{\rm hull}$ has size 
$O(\min\{n_r,m\})$ and ray shooting costs $O(n + \min\{n_r,m\} \log n)$.

\begin{figure}[h]
\centering
\includegraphics[width=0.51\textwidth]{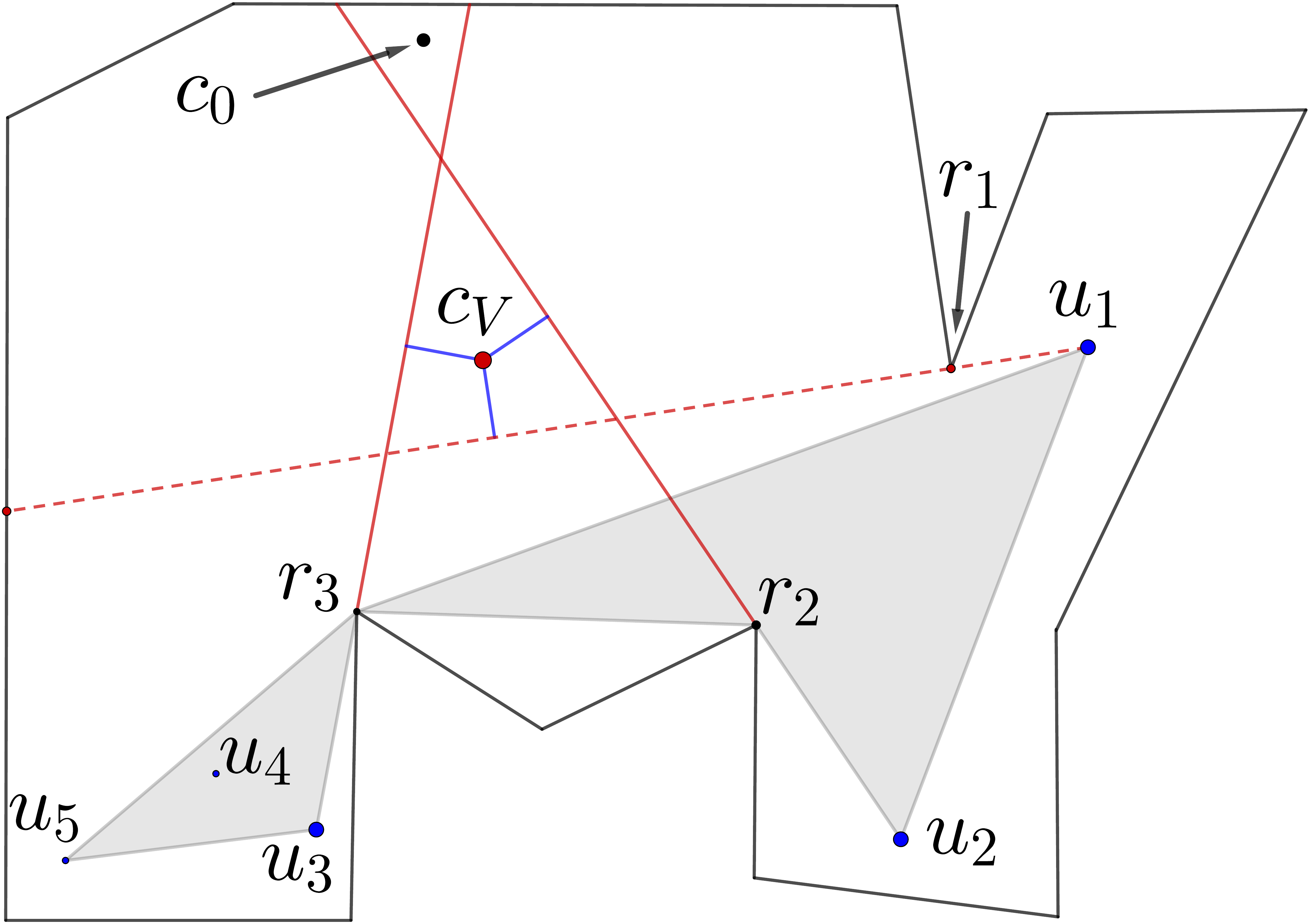}
\caption{
The geodesic convex hull of $U=\{u_1, \ldots, u_5\}$ is shaded grey. ${\cal H}_0$ consists of the two half-polygons $H(u_2,r_2)$ and $H(u_3,r_3)$ (with solid red chords), but misses $H(u_1, r_1)$, which is essential for the visibility center $c_V$.
The point $c_0$ denotes a geodesic center of  ${\cal H}_0$.
}
\label{fig:first_phase_relCH}
\end{figure}

\begin{theorem}
\label{thm:essential-half-polygons}
\newcomment{Suppose the kernel of $U$ is empty. Then the}
geodesic center of ${\cal H}_{\rm hull}$ is the visibility center 
of $U$.
Furthermore ${\cal H}_{\rm hull}$ can be found in time $O(n + m \log (n+m))$
plus the time to find the geodesic center of $O(\min\{n_r,m\})$ half-polygons.
\end{theorem}

\begin{proof}The run-time was analyzed above. 
Consider the visibility center $c_V$.
\newcomment{By assumption, 
$r_V > 0$.} 
We consider the half-polygons $H(u,r) \in {\cal H}_{\rm reflex}$ such that 
$r_V = d(c_V, H(u,r))$.  
By Claim~\ref{claim:2-or-3-vectors} either there 
are three of these half-polygons, \newcomment{$H_1$, $H_2$ and $H_3$,} that uniquely determine $c_V$, or there are two, \newcomment{$H_1$ and $H_2$,} that determine $c_V$.
Then $c_V$ is the geodesic center of $H_i$ 
 $i=1,2,3$ or $i=1,2$ depending on which case we are in.  
 \newcomment{Let $H_i = H(u_i, r_i)$.}

If all the $H_i$'s are in ${\cal H}_0$, we are done. 
We will show that at least two are in ${\cal H}_0$ and the third one (if it exists) is ``caught'' by $c_0$.  See Figure~\ref{fig:first_phase_relCH}.
Let $h_i$ be the chord defining $H_i$ and let $\overline{H}_i$ be the other half-polygon determined by $h_i$.

\begin{claim} 
\label{claim:H0-condition}
If $U$ contains a point in $\overline{H}_i$ then $(u_i,r_i)$ is an edge of $R$ so $H_i \in {\cal H}_0$.
\end{claim}
\begin{proof}
Let $u$ be a point in $\overline{H}_i$.  Observe that 
$\pi(u_i,u)$ contains the segment $u_i r_i$. 
Thus $r_i$ is a vertex of $R$.
Furthermore $u_i r_i$ is an edge of $R$.  (Note that $H_i$ is extreme at $r$ since we picked it from ${\cal H}_{\rm reflex}$.)
Thus $H_i$ is in ${\cal H}_0$.
\end{proof}

\begin{claim}
\label{claim:at_least_two}
At least two of the ${H}_i$'s lie in ${\cal H}_0$.
\end{claim}
\begin{proof}
\newcomment{
First observe  that if two of the half-polygons are disjoint, say $H_i$ and $H_j$, then they lie in ${\cal H}_0$, because 
$u_i \in H_i$ implies $u_i \in \overline{H}_j$ so by Claim~\ref{claim:H0-condition}, $H_i \in {\cal H}_0$, and 
symmetrically, $H_j \in {\cal H}_0$.}

\newcomment{We separate the proof into cases depending on the number of $H_i$'s. 
If there are two then they must be 
disjoint otherwise a point in their intersection would be a visibility center with visibility radius $r_V=0$.
Then by the above observation, they are both in ${\cal H}_0$}. 

It remains to consider the case of three half-polygons.
\newcomment{If two are disjoint, we are done,
so suppose each pair $H_i, H_j$ intersects.}  Then the three 
chords $h_i$
form a triangle.  Furthermore, since $\bigcap {\overline{H}_i}$ is non-empty (it contains $c_V$), the inside of the triangle is $ \bigcap {\overline{H}_i}$.  
Now suppose $H_1 \notin {\cal H}_0$.  
Then by Claim~\ref{claim:H0-condition}, $u_2, u_3 \in H_1$. 
This implies (see Figure~\ref{fig:first_phase_relCH}) that $u_2 \in \overline{H}_3$ and $u_3 \in \overline{H}_2$, so by   Claim~\ref{claim:H0-condition}, $H_2$ and $H_3$ are in ${\cal H}_0$.
\end{proof}

\newcomment{We now complete the proof of the theorem.
We only need to consider the case of three $H_i$'s, where one of them, say $H_1$, is not in ${\cal H}_0$.
}
Our goal is to show that 
$c_0$, the geodesic center of ${\cal H}_0$, lies in $\overline{H}_1$ and thus $H_1$ is in ${\cal H}_{\rm hull}$. 
Let 
$X = \{x \in P : d(x,{H}_2) \le r_V {\rm \ and\ } d(x,{H}_3) \le r_V \}$.  Observe that $c_0 \in X$ (because the radius is non-increasing as we eliminate half-polygons).  Now, $c_V$ is the unique point within distance $r_V$ of the half-polygons ${H}_1$, ${H}_2$ and ${H}_3$.  If $c_0 \in H_1$, then $c_0$'s distance to ${H}_1$  would be 0 
\newcomment{which contradicts the uniqueness property of $c_V$.}
Thus $c_0 \in \overline{H}_1$. 
By the same reasoning as in Claim~\ref{claim:H0-condition}, this implies that $u_1 r_1$ is an edge of $R'$, the geodesic convex hull of $U \cup \{c_0\}$.  Thus 
$H_1$ is in ${\cal H}_{\rm hull}$ by definition of ${\cal H}_{\rm hull}$.
\end{proof}

\section{The Geodesic Center of Half-Polygons}\label{section:half-polygon-center}

In this section, we 
give an algorithm to find
the geodesic center of a set $\mathcal{H}$ of $k$ half-polygons inside an $n$-vertex polygon $P$
\changed{with run time $O((n+k) \log (n+k))$}.
\fchanged{(Note that although we say ``the'' geodesic center, 
it need not be unique, see Claim~\ref{claim:2-or-3-vectors}.)} %
We preprocess by sorting the half-polygons in cyclic order of their first endpoints around $\partial P$ in time $O(k \log k)$.
\changed{We assume that no half-polygon contains another---such irrelevant non-minimal half-polygons can be detected from the sorted order and discarded.}
\fchanged{We also make the general position assumption that no point in $P$ 
has equal non-zero distances 
to more than 
a constant number of 
half-polygons of $\cal H$.}

We follow the  approach that  Pollack et al.~\cite{pollack_sharir} used to find the geodesic center of the vertices of a polygon.  Many  steps of their algorithm rely, in turn, on search algorithms of Megiddo's~\cite{megiddo_linear}.

The main ingredient of the algorithm 
is a linear time \defn{chord oracle} that, given a chord $K=ab$ of the polygon,
finds the \defn{relative geodesic center}, 
$c_K$ (the optimum center point restricted to points on the chord), 
and tells us which side of the chord contains the center. 
We must completely redo the chord oracle %
in order to handle paths to half-polygons instead of vertices, but the main steps are the same.  
Our chord oracle runs in time $O(n+k)$. 
The chord oracle of Pollack et al.~was
used as a black box in subsequent faster algorithms~\cite{linear_time_geodesic}, so we imagine that our version will be an ingredient in any faster algorithm for the geodesic center of half-polygons. %

Using the chord oracle, we again follow the approach of Pollack et al.~to find the geodesic center.
The total run time is $O((n+k) \log (n+k))$.

We 
give a road-map
for the remainder of this section, listing the main steps, \newcomment{which are the same as those of Pollack et al.,}
and 
highlighting the parts that
\newcomment{we must rework.} 

\medskip
\noindent{\bf \S~\ref{section:chord_oracle} A Linear Time Chord Oracle}
\begin{enumerate}
\item 
Test a candidate center point.  \changed{Given the relative geodesic center $c_K$ on chord $K = ab$, is the geodesic center to the left of right of chord $K$?
This test reduces the chord oracle to finding the relative geodesic center, which is done via the following steps.}

\item Find shortest paths from $a$ and from $b$ to all half-polygons. 
The details of this step are novel, because we need shortest paths to half-polygons rather than vertices.  
\item Find a linear number of simple functions defined on $K$ whose upper envelope is the geodesic radius function.
We must redo this from the ground up. 
\item Find the relative center on $K$ (the point that minimizes the geodesic radius function) using Megiddo's technique.
\end{enumerate}

\medskip
\noindent
{\bf \S~\ref{section:center_using_oracle} Finding the Geodesic Center of Half-Polygons}
\begin{enumerate}
\item Use the chord oracle to find
a region of $P$ that contains the center and such that for any half-polygon $H \in \cal H$, all geodesic paths from the region to $H$ are combinatorially the same.
We give a more modern version of this step using epsilon nets.
\item Solve the resulting Euclidean 
\changed{problem of finding}
a smallest disk that contains given disks and
\newestchanged{intersects given half-planes.}
This is new because of the condition about intersecting
\newestchanged{half-planes}.
\end{enumerate}

\subsection{A Linear Time Chord Oracle}
\label{section:chord_oracle}

In this section we give a linear time chord oracle. 
Given a chord $K=ab$ the chord oracle tells us whether the geodesic center of $\cal H$ 
lies to the left, to the right, or on the chord $K$.  
It does this by first finding the relative geodesic center %
$c_K = \argmin \{r(x,{\cal H}): x \in K\}$, together with the half-polygons
that are farthest from $c_K$
\changed{and the first segments of the shortest paths from $x$ to those farthest half-polygons.}
From this information, 
we can identify which side of $K$ contains the geodesic center $c$ in the same way as Pollack 
et al.~by testing the vectors of the first segments of the shortest paths from $c_K$ to its furthest half-polygons.  This test is described in Subsection~\ref{sec:testing_center}. %

The chord oracle thus reduces to the problem of finding the relative geodesic center and its farthest half-polygons.  
The main idea here is to capture the geodesic radius function along the chord  (i.e., the function $r(x,{\cal H})$ for $x \in K$) as the upper envelope of a linear number of \changed{easy-to-compute convex} functions defined on overlapping subintervals of $K$. 
In order to find the %
functions (Section~\ref{section:functions_to_capture}) we first 
compute shortest paths from $a$ and from $b$ to all the half-polygons (Section~\ref{section:shortest_path_tree_half_polygons}).  Finally we apply Megiddo's techniques (Section~\ref{sec:find-rel-center}) to find the point $c_K$ on $K$ that minimizes the geodesic radius function.

\subsubsection{Testing a Candidate Center Point}
\label{sec:testing_center}

\changed{In this section we show how to test in constant time
whether a candidate point is a geodesic center, or relative center, and if not,
in which direction the  center lies. 
The basic idea is that a local optimum is a global optimum, so a local test suffices.
In more detail, the input is a point $x$ on a chord $K=ab$ together with 
\fchanged{its geodesic radius $r(x, {\cal H})$ and}
the first segments of the shortest paths from $x$ to its farthest half-polygons.  The goal is to test in constant time:
(1) whether $x$ is a relative geodesic center of $K$, and if not, which direction to go on $K$ to reach a relative center; and (2) if $x$ is a relative geodesic center,
whether $x$ is a  geodesic center of $P$, and if not, which side of $K$ contains a geodesic center of $P$.
These tests are illustrated in Figure~\ref{fig:master_center_conditions}.}
\fchanged{Note that if $r(x, {\cal H})$ is zero, then 
$x$ is a geodesic center and no further work is required.}

\begin{figure}[h]
\centering
\includegraphics[width=.95\textwidth]{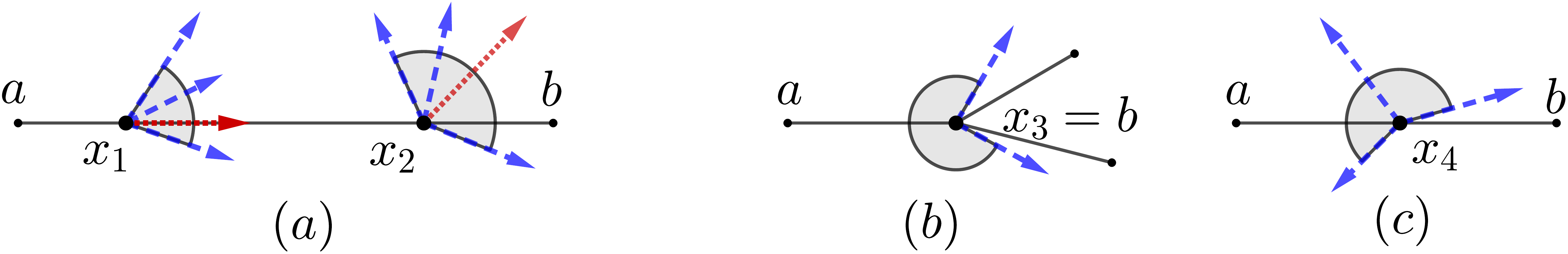}
\caption{
Points $x_i$ on chord $K=ab$ with directions of paths to farthest half-polygons in dashed blue,
wedge $W_{\alpha}$ 
shaded, 
and
the direction for improvement %
in dotted red.
(a) $x_1$ is not a relative center of $K$, whereas
 $x_2$ is a relative center but not the geodesic center of $P$. 
(b,c) $x_3$ (a reflex vertex of the polygon) and $x_4$ are geodesic centers.
}
\label{fig:master_center_conditions}
\end{figure}

\changed{
The tests are accomplished via the following lemma, 
which is analogous to Lemmas 2 and 3 of Pollack et al.~\cite{pollack_sharir}.
}

\begin{lemma}
\label{lemma:test_for_global_center}
Let $x$ be a point on chord $K$, and 
let $V$ be the vectors of the first segments of the shortest paths from $x$ to its farthest half-polygons
\newestchanged{$\cal F$}. 
\revised{Let $\alpha$ be the smallest angle of a wedge $W_\alpha$ with apex $x$ that contains all the vectors of $V$ and such that $W_\alpha$, restricted to a small neighbourhood of $x$, is contained in $P$.}
\begin{enumerate}
\item {\bf Location of the relative center.}
Let $L$ be the line through $x$ perpendicular to $K$. 
If one of the open half-planes determined by $L$ contains $W_\alpha$, 
then $x$ is not a relative center, and all relative centers lie on that side of $L$. 
Otherwise, $x$ is a relative center.

\item {\bf Location of the center.}
Now suppose that $x$ is a relative geodesic center.
\revised{If $\alpha < \pi$ then $x$ is not a geodesic center, and all geodesic centers lie on the side of $K$ that contains the 
\newestchanged{ray bisecting the angle of $W_\alpha$.}
If $\alpha > \pi$, then $x$ is the unique geodesic center, and furthermore, $x$ is determined by two or three vectors of $V$---the two that bound $W_\alpha$, plus one inside $W_\alpha$ unless $x$ is on the boundary of $P$.  Finally, if $\alpha = \pi$ then $x$ is a geodesic center (though not necessarily unique), and furthermore, $x$ is determined by the two vectors of $V$ that bound $W_\alpha$. } %
\end{enumerate}
\end{lemma}

\newestchanged{
\begin{proof}
We prove the two parts separately.
\begin{enumerate}
\item 
Suppose $W_\alpha$ lies in an open half-plane determined by $L$ (say, the left side of $L$). 
Then moving $x$ an epsilon distance left along $K$ gives a point with smaller geodesic radius since the distance to any  half-polygon in $\cal F$ decreases, and no other half-polygon becomes a farthest half-polygon.  
Therefore $x$ is not a relative center. Furthermore, because the geodesic radius function $r(x,{\cal H})$ is convex on $K$ (by Corollary~\ref{cor:geodesic-convexity}), the relative center lies to the left on $K$. 

Next suppose $W_\alpha$ does not lie in an open half-plane of $L$. Then any epsilon movement of $x$ along $K$ increases the distance to some half-polygon in $\cal F$, so $x$ is a local minimum on $K$ and therefore $x$ is the relative center (again using the face that the geodesic radius function is convex on $K$).

\item Suppose $\alpha < \pi$.  Let $b$ be the ray that bisects $W_\alpha$.  Moving $x$ an epsilon distance along $b$ gives a point $x'$ with smaller geodesic radius.  Therefore $x$ is not the center. 
Next we prove that the center $c$ lies on the side of $K$ that contains $b$.
Suppose not. Consider the geodesic $\pi(c,x')$.  By Corollary~\ref{cor:geodesic-convexity}, the geodesic radius function is convex on $\pi(c,x')$.  But then the point where the geodesic crosses $K$ has a smaller geodesic radius than $x$, a contradiction to $x$ being the relative center.

Next suppose $\alpha > \pi$. Let $v_1$ and $v_2$ be the two vectors that bound $W_\alpha$. 
If $x$ is on the boundary of $P$ it must be at a reflex vertex of $P$.
Otherwise, 
since no smaller wedge contains $V$, there must be a third vector $v_3$ in $V$, making an angle $< \pi$ with each of $v_1$ and $v_2$.  In either case ($x$ on the boundary of $P$ or not) 
any epsilon movement of $x$ in $P$ increases the distance to the half-polygon corresponding to one of the $v_i$'s. 
Thus $x$ is a local minimum in $P$ and (by geodesic convexity of the radius function) $x$ is the center.  Furthermore, $x$ is determined by $v_1$ and $v_2$---and $v_3$ if $x$ is interior to $P$. 

Finally, suppose $\alpha = \pi$. As in the previous case, $x$ is a geodesic center 
and is determined by the two vectors $v_1$ and $v_2$ of $V$ that bound $W_\alpha$.  Furthermore, $x$ is unique unless the two corresponding half-polygons 
have parallel defining chords, and $v_1$ and $v_2$ reach those chords at right angles.  In this case the set of geodesic centers consists of a line segment through $x$ parallel to the chords.

\end{enumerate}
\end{proof}
}

\subsubsection{Shortest Paths to Half-Polygons}
\label{section:shortest_path_tree_half_polygons}

In this section we give a linear time algorithm to find the shortest path tree from point $a$ on the polygon boundary to all the half-polygons $\cal H$.  Recall that each half-polygon is specified by an ordered pair of endpoints on $\partial P$, and  the half-polygons are sorted in clockwise cyclic order by their first endpoints.  From this, we identify  the half-polygons that contain $a$, and we  discard them---their distance from $a$ is 0.
Let $H_1, \ldots, H_{k'}$ be the remaining half-polygons where $H_i$ is bounded by endpoints $p_i q_i$, and the $H_i$'s are sorted by $p_i$, starting at $a$.  

The idea is to first find the shortest path map $T_a$ from $a$ to the set consisting of the polygon vertices and  the points $p_i$ and $q_i$. Recall that the shortest path map is an augmentation of the shortest path tree that partitions the polygon into triangular regions in which the shortest path from $a$ is combinatorially the same (see Figure~\ref{fig:shortest-paths}).
The shortest path map can be found in linear time~\cite{SPT_linear}.
Note that $T_a$ is embedded in the plane (none of its edges cross) 
and the ordering of its leaves matches their ordering on $\partial P$.
Our algorithm will traverse $T_a$ in depth-first order, and visit the triangular regions along the way.

Our plan is to augment $T_a$ to a shortest path tree $\bar T_a$ that includes the shortest paths from $a$ to each half-polygon $H_i$.  Note that $\bar T_a$ is again an embedded ordered tree.
We can find $\pi(a, H_i)$ by examining the regions of the shortest path map intersected by $p_i q_i$. These lie in 
the \emph{funnel} 
between the shortest paths 
$\pi(a,p_i)$ and $\pi(a, q_i)$. 
Note that edges of the shortest path map $T_a$ may cross the chord $p_i q_i$.  Also, 
the funnels for different half-polygons may overlap.  The key to making the search efficient is the following lemma:

\begin{lemma}
The ordering $H_1, H_2, \ldots, H_{k'}$ matches the ordering of the paths $\pi(a,H_i)$ in the tree $\bar T_a$.  
\end{lemma}
\begin{proof} 
Consider two half-polygons $H_i = p_iq_i$ and $H_j = p_jq_j$, with $i < j$.  We prove that $\pi(a,H_i)$ comes before $\pi(a,H_j)$ in $\bar T_a$.
If $H_i$ and $H_j$ are disjoint, the result is immediate since the corresponding funnels do not overlap.  Otherwise (because neither half-polygon is contained in the other) $p_i q_i$ and $p_j q_j$ must intersect, say at point $x$. See Figure~\ref{fig:shortest-paths}.
Let $t_i$ and $t_j$ be the terminal points of the paths $\pi(a,H_i)$ and $\pi(a,H_j)$, respectively.  
If $t_i$ lies in $p_ix$ and $t_j$ lies in $xq_j$ then the result follows since $t_i$ and $t_j$ lie in order on the boundary of the truncated polygon formed by removing $H_i$ and $H_j$.
So suppose that $t_j$ lies in $p_jx$ (the other case is symmetric).  Then $\pi(a,t_j)$ crosses $p_i q_i$ at a point $z$ in $p_i x$. 
From $z$ to  \newcomment{$t_j$} the path $\pi(a,t_j)$ lies inside the 
\newcomment{
cone
with apex $x$ bounded by the rays from $x$ through $z$ and from $x$ through $t_j$.
Within that cone, the path only turns left.}  
The angle $\alpha_j$ at $t_j$ is $\ge 90^\circ$ (it may be $> 90^\circ$ if $t_j = p_j$), which implies that the angle $\alpha_i$ at $z$ is $> 90^\circ$.  Therefore $t_i$ lies to the left of $z$, as required.
\end{proof}

\begin{figure}[h]
\centering
\includegraphics[width=0.59\textwidth]{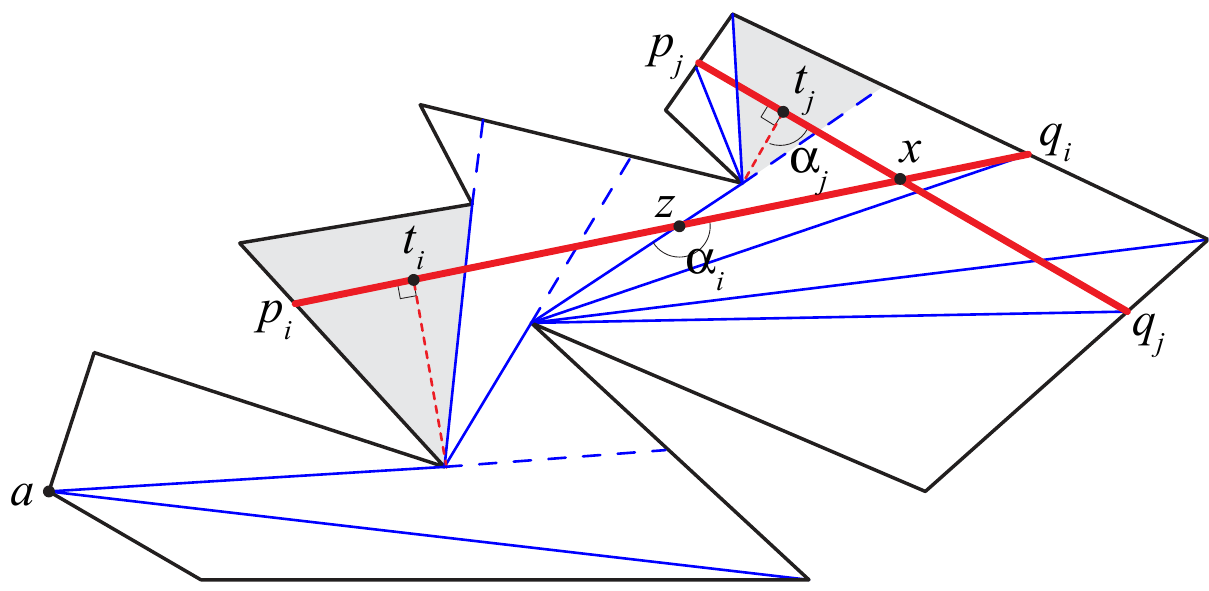}
\caption{The shortest path map $T_a$ (thin blue) and the augmentation (dashed red) to include shortest paths to the
two half-polygons bounded by chords $p_i q_i$  and $p_j q_j$ (thick red).
}
\label{fig:shortest-paths}
\end{figure}

Based on the Lemma, the algorithm
traverses the regions of the shortest path map $T_a$ in depth first search order, and traverses the half-polygons $H_i$ in order $i=1,2, \ldots, k'$. 
It is easy to test if one region contains the shortest path to $H_i$ 
(either to $p_i$, or to $q_i$, or reaching an internal point of $p_i q_i$ at a right angle); 
if it does, we increment $i$, and otherwise we proceed to the next region.
The total time is $O(n+k)$.

\subsubsection{Functions to Capture the Distance to Farthest Half-Polygons}
\label{section:functions_to_capture}

\changed{In this section we capture the geodesic radius function for points on a chord $K=ab$ as the upper envelope of functions defined on overlapping subintervals of $K$. 
Besides extending the method of Pollack et al.~\cite{pollack_sharir} to deal with half-polygons (rather than vertices), 
our aim is to give a clearer and easier-to-verify presentation.}

In  more detail, we
give a linear time algorithm to find
a linear number of 
\changed{easy-to-compute convex} functions defined on the chord $K=ab$ whose upper envelope is the geodesic radius function $r(x,{\cal H})$ for $x \in K$.
Specifically, a \defn{coarse cover} is 
a set of %
\changed{triples $(I,f,H)$}
where: 
\begin{enumerate}
\squeezelist
\item 
$I$ is a subinterval of $K$, $f$ is a function defined on domain $I$, 
\fchanged{and $H \in {\cal H}$.}

\item 
$f(x) = d(x,H)$
\fchanged{for all $x \in I$}, and
$f$ 
\newestchanged{
has one of the following forms:
\begin{itemize}
\squeezelist
\item $f(x) = 0$.
\item $f(x) = d_2(x,v) + \kappa$ where $d_2$ is Euclidean distance, $\kappa$ is a constant,
$v$ is 
a vertex of $P$, and 
the segment $xv$ is the first segment of the path $\pi(x,H)$.

\item $f(x) = d_2(x, {\bar h})$, where $d_2$ is Euclidean distance, $\bar h$ is the line through the defining chord of $H$,  and the path $\pi(x,H)$ is the straight line segment from $x$ to $\bar h$ (meeting $\bar h$ at right angles).
\end{itemize}
}

\item
\newestchanged{For any point $x \in K$  and any half-polygon $H$ that is farthest from $x$, there is a triple $(I,f,H)$ in the coarse cover with $x \in I$---with the exception that if two triples have identical $I$ and %
\attention{identical $f=d_2(x,v) + \kappa$},
then we may eliminate one of them.

In particular, this implies that the upper envelope of the functions is the geodesic radius, i.e., 
for any $x \in K$, the
maximum of $f(x)$ over intervals $I$ containing $x$ is equal to $r(x,{\cal H})$.
}
\end{enumerate}

For intuition, see Figure~\ref{fig:intervals-functions}, 
which shows several intervals and their associated functions. 
\fchanged{We will find the elements of the coarse cover separately for the two pieces of the polygon on each side of $K$, and then take the union of the two sets.}
In this section we visualize $K$ as horizontal and deal with the upper piece of the polygon.

\begin{figure}[h]
  \centering
\includegraphics[width=0.6\textwidth]{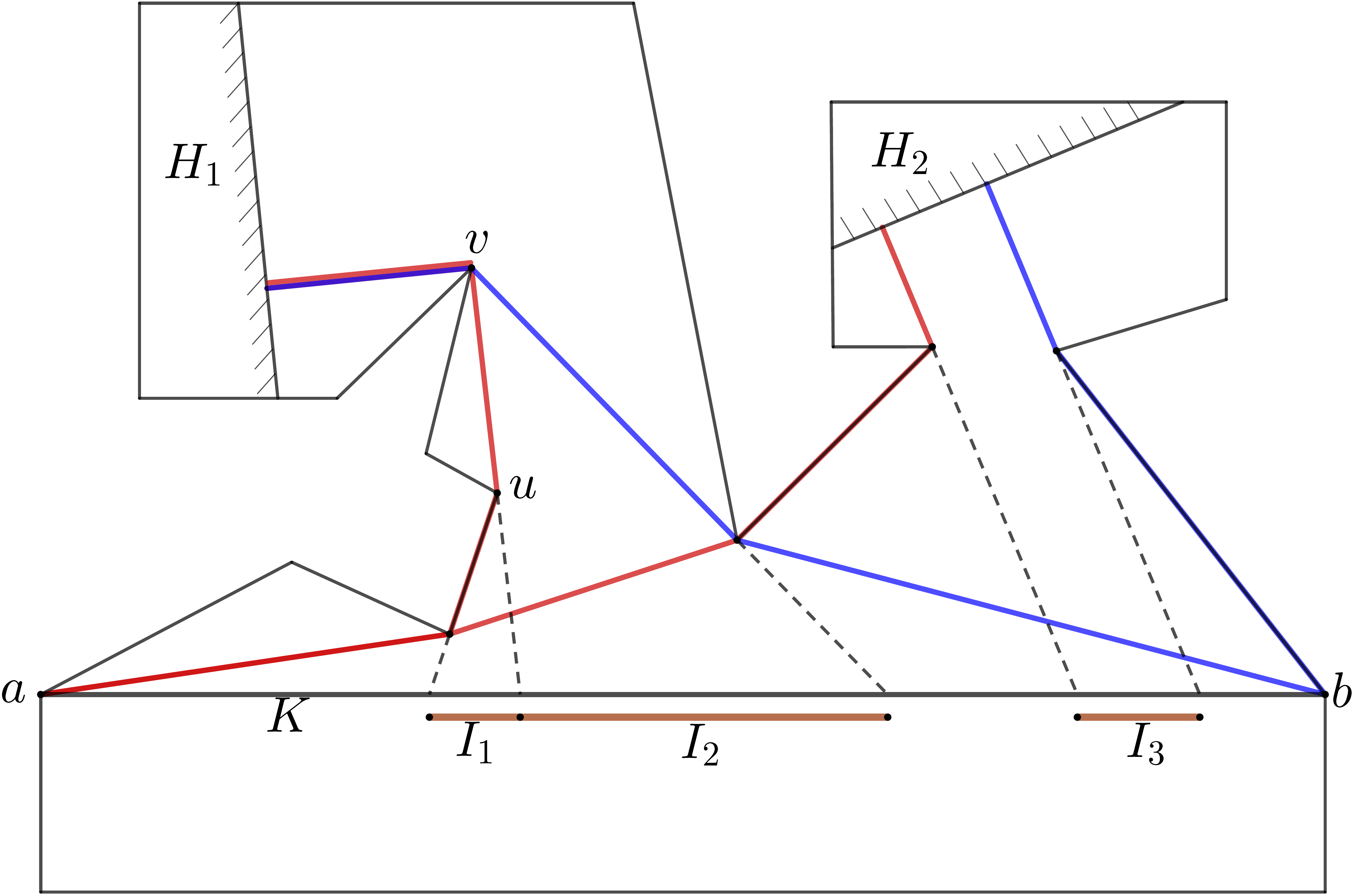}
\caption{
An illustration of functions and intervals. For $x$ in interval $I_1$, $d(x,H_1) = d_2(x,u) + \kappa_1$. For $x$ in $I_2$, $d(x,H_1) = d_2(x,v) + \kappa_2$.  For $x$ in $I_3$, $d(x,H_2) = d_2(x,H_2)$.
}
\label{fig:intervals-functions}
\end{figure}

\medskip
\noindent{\bf A large coarse cover.}
We first describe a
coarse cover of $O(nk)$ triples 
and then show how to 
reduce to linear size.
Consider a half-polygon $H$ with defining chord $h$.  Suppose first that $K$ does not intersect $H$, i.e., $a$ and $b$ lie outside $H$. 
All shortest paths from points on $K$ to $H$ lie in the \defn{funnel} $Y(H)$ which is a subpolygon bounded by the chord $K$, the path $\pi(a,H)$ (which is a path in ${\bar T}_a$), the path $\pi(b,H)$ (in ${\bar T}_b$), and the segment along $h$ between the terminals of those two paths.
See Figure~\ref{fig:funnels}. 
If the paths $\pi(a,H)$ and $\pi(b,H)$ are disjoint then they are both reflex paths and all vertices on the paths are visible from $K$ (see Figure~\ref{fig:funnels}(a)).  Otherwise, the paths are reflex and visible from $K$ until they reach the first common vertex $u$, and then they have a common subpath from $u$ to $H$ that is not visible from $K$ (see Figure~\ref{fig:funnels}(b)).

\fchanged{Before describing how to obtain triples of the coarse cover from $Y(H)$, we first consider the case when $K$ intersects $H$, i.e., $a$ or $b$ lies inside $H$.
If both $a$ and $b$ are inside $H$, then
we add the triple $(I=ab,f=0,H)$ to the coarse cover. 
If $b$ is outside but $a$ is inside (the other case is symmetric), then $h$ and $K$ intersect at a point $p$. 
If $\pi(b,H)$ reaches $H$ below $K$, then $H$ will be handled when we deal with the piece of the polygon below $K$. 
Otherwise (see Figure~\ref{fig:funnels}(c)) we
add the triple $(I=ap, f=0, H)$ to the coarse cover, and we deal with the $pb$ portion of the chord as in the general case above but modifying the funnel $Y(H)$ so that the path $\pi(a,H)$ is replaced by $p$. 
}

\begin{figure}
    \centering
    \includegraphics[width=.95\textwidth]{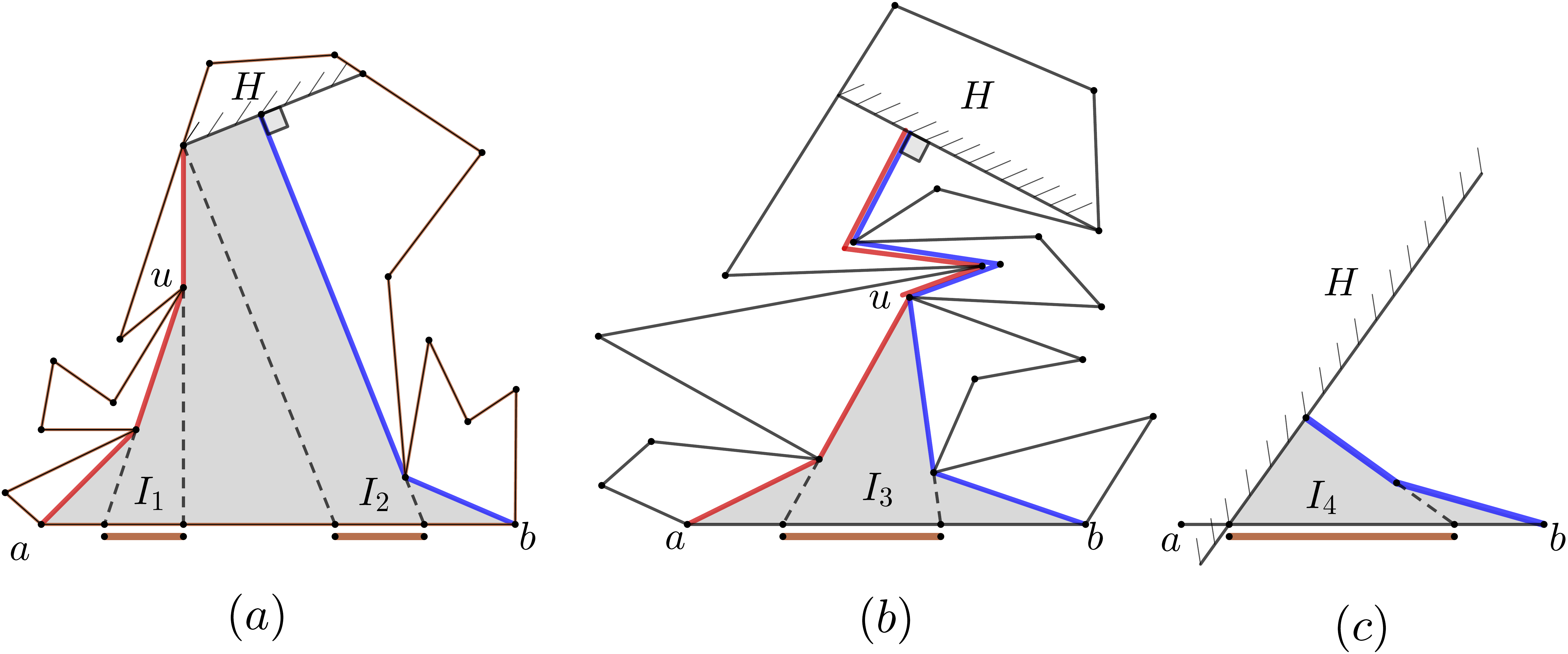}
    \caption{The funnel $Y(H)$ (shaded) and its shortest path map which is bounded by edge extensions (dashed segments). (a) 
    The case of disjoint paths $\pi(a,H)$ and $\pi(b,H)$.
    For $x \in I_1$, $d(x,H) = d_2(x,u) + \kappa$.  For $x \in I_2$, $d(x,H) = d_2(x,{\bar h})$.
    (b) The case of paths that meet at vertex $u$. 
    (c) The case where $a \in H$.}
    \label{fig:funnels}
\end{figure}

Each funnel $Y(H)$ can be partitioned into its \defn{shortest path map 
$M(H)$}
where two points are in the same region of 
$M(H)$
if their paths to $H$ are combinatorially the same. 
(We consider a path that arrives at an endpoint of $h$ and a path that arrives at an interior point of $h$ to be combinatorially different.) 
Observe that boundaries of the regions of $M(H)$ are extensions of tree edges plus lines perpendicular to $h$. 
See Figure~\ref{fig:funnels}.
The regions of 
$M(H)$
are triangles, plus possibly one trapezoid. 
A triangle region has 
a base segment $I \subseteq K$, and 
an apex vertex $u$ on $\pi(a,H)$ [or $\pi(b,H)$];  the shortest path from any point $x \in I$ to $H$ consists of the line segment $xu$ plus the path in ${\bar T}_a$ [or ${\bar T}_b$] from $u$ to $H$, so $d(x,H) = d_2(x,u) + \kappa$ where  $\kappa$ is the tree distance from $u$ to the leaf corresponding to $H$.  
A trapezoid region has 
a base segment $I \subseteq K$, and 
two sides orthogonal to $h$; the shortest path from any point $x \in I$ to $H$ consists of the line segment orthogonal to $h$ from $x$ to $h$,
so $d(x,H) = d_2(x,{\bar h})$ where $\bar h$ is the line through $h$. 
Thus each region of $M(H)$ gives rise to a triple $(I,f,H)$ satisfying properties (1) and (2) of a coarse cover.  

\changed{We claim that the set of triples defined above, i.e., 
all the triples 
defined from $Y(H)$ together with the special triples when $H$ intersects $K$, form a coarse cover. Properties (1) and (2) are satisfied, and property (3) is satisfied  
because}
we have captured all shortest paths from $x$ to $H$ for all $x \in K$ and all half-polygons $H$.
Since each $Y(H)$ has size $O(n)$, this coarse cover has size $O(nk)$.

\medskip
\noindent{\bf Intuition for a linear-size coarse cover.}
The secret to reducing the size of the coarse cover is to observe that if the funnels 
\newestchanged{for some
half-polygons ${\cal H}' \subseteq {\cal H}$}
share an edge $uv$ of ${\bar T}_a$ with $u$ closer to the root, and both $u$ and $v$ visible from $K$, then their shortest path maps share the same triangle with apex $u$, 
\newestchanged{base $I$},
and sides bounded by the %
\changed{extension of the edge from $v$ to $u$ and the extension of the edge from $u$ to its parent in ${\bar T}_a$ (see Figure~\ref{fig:coarse-cover}(a)).}
\newestchanged{In this case, we claim that for this triangle, we only need a coarse cover element for one of the half-polygons in ${\cal H}'$, specifically, for one 
half-polygon that has the maximum distance from $v$ in the tree ${\bar T}_a$.  This is because only half-polygons farthest from $v$ matter, and furthermore,
we need not keep more than one half-polygon that has the maximum distance because the interval $I$ and the function $f(x) = d_2(x,u) + \kappa$ are the same.} 
We first specify the coarse cover precisely and then prove correctness, which makes the above observation formal.

\medskip
\noindent{\bf Definitions.}
Let ${\bar T}_a$ and ${\bar T}_b$ be directed from root to leaves.
For any node $v$ in $\bar T_a$ %
define \defn{$\ell_a(v)$} to
be the maximum length of a directed path in $\bar T_a$ from $v$ to a  leaf node %
representing a terminal point on some half-polygon, and define 
\defn{$F_a(v)$} to be that farthest half-polygon
\newestchanged{(breaking ties arbitrarily)}.  
Define functions \defn{$\ell_b$} and 
\defn{$F_b$} similarly.
We can compute these functions in linear time in leaf-to-root order. In particular,  
we  compute $\ell_a(u)$ for the nodes $u$ of $\bar T_a$ as follows.  Initialize $\ell_a(u)$ to 0 if $u$ represents a terminal point of a half-polygon chord, and  to $-\infty$ otherwise. Then from leaf-to-root order, update  $\ell_a(u)$ to $\max   \{\ell_a(u),  \max  \{|u v| + \ell_a(v) :  v$ a child of $u \}\}$. 
We can compute $\ell_b(u)$ similarly.  The runtime is $O(n+k)$.

Define \defn{$p_a(u)$} and \defn{$p_b(u)$} to be the parents of node $u$ in $\bar T_a$ and $\bar T_b$, respectively.
As noted by 
Pollack et al.~\cite{pollack_sharir}, 
a vertex $u$ is visible from some point on $K$ if and only if $p_a(u) \ne p_b(u)$. 
Furthermore, we note that if $u$ is visible from some point on $K$, then extending the edge from $u$ through $p_a(u)$ reaches a point
\defn{$x_a(u)$}
on $K$ from which $u$ is visible.
Similarly, extending the edge from $u$ through $p_b(u)$ reaches a point
\defn{$x_b(u)$}
on $K$ from which $u$ is visible.

\changed{
In defining the  shortest path map $M(H)$, we added boundary lines orthogonal to the defining chord $h$ at the terminals of the paths $\pi(a,H)$ and $\pi(b,H)$.  If a path terminates at an internal point of $h$ then the last edge of the path is orthogonal to $h$, and the boundary line extends the last edge.
In order to avoid special cases, it will be convenient if all boundary lines are extensions of tree edges, i.e.,   
to assume that even the paths that terminate at endpoints of $h$ end with a segment orthogonal to $H$.  We add $0$-length segments to the trees
${\bar T}_a$ and ${\bar T}_b$ to make this true. The extension of such a $0$-length edge is orthogonal to $H$.
Note that it is possible that both $\pi(a,H)$ and $\pi(b,H)$ terminate at the same endpoint of $h$, in which case the added 0-length segment is common to both trees, so we regard the terminal point of the paths as \emph{not} visible from $K$. 
See Figure~\ref{fig:coarse-cover}(c).
(The other endpoint of the 0-length segment may or may not be visible.)}

\medskip
\noindent{\bf The coarse cover $\cal T$.}
Define \defn{$\cal T$} to have elements
of the following %
\fchanged{four}
types.  See Figure~\ref{fig:coarse-cover}.

\begin{enumerate}
\squeezelist
\setcounter{enumi}{-1}
\item 
\fchanged{For each half-polygon $H$ that intersects $K$ there is a coarse cover element $(I,f,H)$ where $I = K \cap H$ and $f(x) = 0$.} 

\item For each edge $(u,v)$ in ${\bar T}_a$ where $u = p_a(v)$,
\changed{$u \ne a$}, and $u$ and $v$ are both visible from $K$, there is an associated 
{\bf $a$-side triangle} that has apex $u$ and base $I = [x_a(u), x_a(v)] \subseteq K$. 
The associated coarse cover element is $(I,f,H)$ where $H = F_a(v)$ and 
\attention{$f(x) = d_2(x,u) + |uv| + \ell_a(v)$.}
Define {\bf $b$-side triangles} and their associated coarse cover elements symmetrically.

\item For each edge $(u,v)$ that is common to ${\bar T}_a$ and ${\bar T}_b$ where $u$ is visible from $K$ and $v$ is not 
(i.e., $u = p_a(v) = p_b(v)$) there is an associated 
{\bf central triangle} that has apex $u$ and 
base $I = [x_a(u), x_b(u)] \subseteq K$.
The associated coarse cover element is $(I,f,H)$ where $H = F_a(v)= F_b(v)$ and
\attention{$f(x) = d_2(x,u) + |uv| + \ell_a(v)$.}

\item For each half-polygon $H$ such that the terminal points \defn{$t(a,H)$} of $\pi(a,H)$ and \defn{$t(b,H)$} of $\pi(b,H)$ are distinct, 
there is a {\bf central trapezoid}
with base $I \subseteq K$ bounded by the 
two lines perpendicular to $h$ emanating from $t(a,H)$ and $t(b,H)$---these lines are the extensions of the (possibly $0$-length) last edges of the paths.
The associated coarse cover element is $(I,f,H)$ where $H$ is the given half-polygon and $f(x) = d_2(x,{\bar h})$ where $\bar h$ is the line through the defining chord of $H$. 

\end{enumerate}

\begin{figure}[htb]
    \centering
    \includegraphics[width=\textwidth]{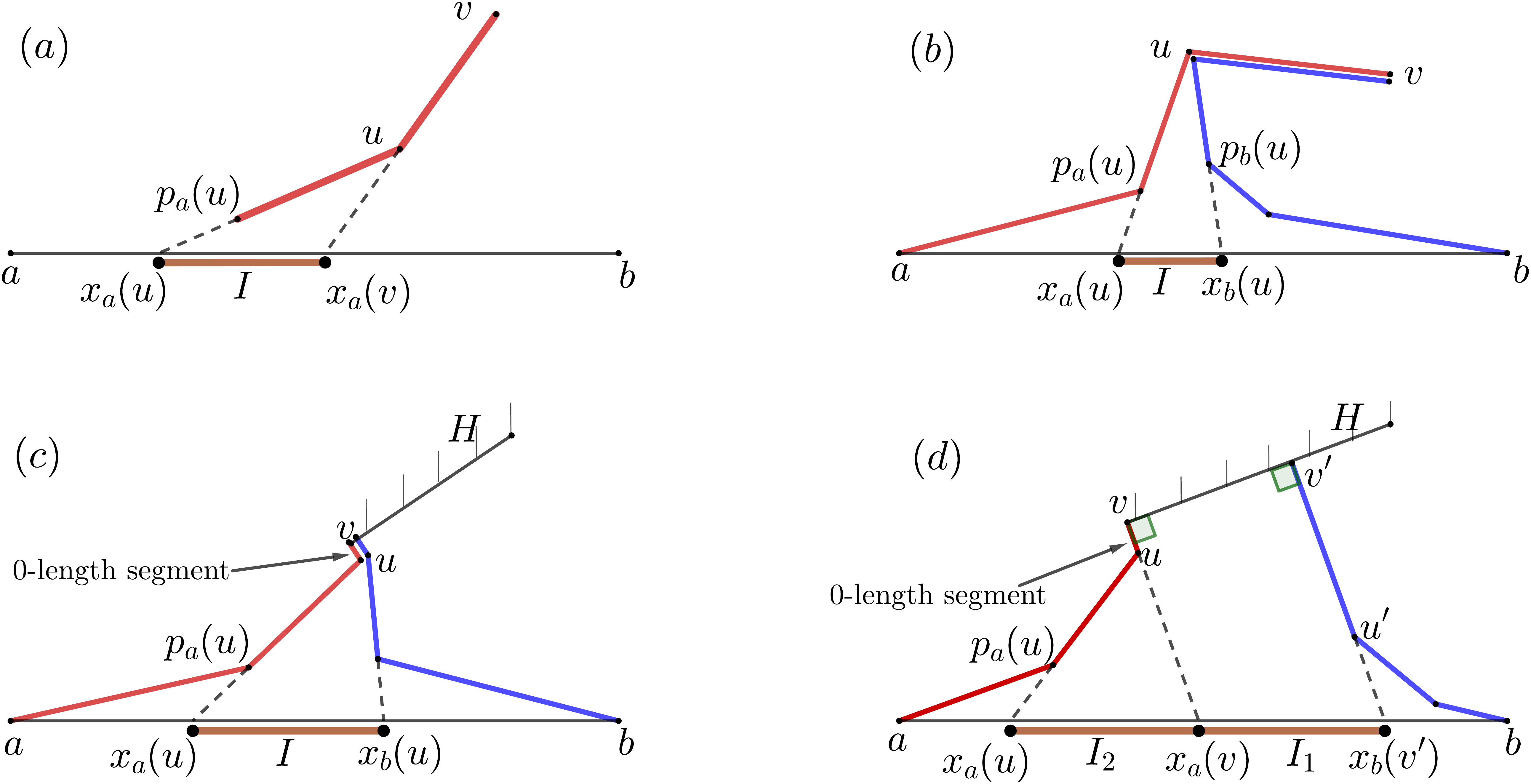}
    \caption{
    Elements of the coarse cover: (a) $I$ is the base of an $a$-side triangle associated with edge $uv$; (b) $I$ is the base of a central triangle associated with edge $uv$; 
    (c) $I$ is the base of a central triangle associated with $0$-length edge $uv$;
    (d) $I_1$ is the base of a central trapezoid associated with terminal points $v$ and $v'$ on $H$; $I_2$ is the base of an $a$-side triangle  associated with $0$-length edge $uv$.
    }
    \label{fig:coarse-cover}
\end{figure}

Note that we include $0$-length edges in cases 1 and 2 above.
Altogether, $\cal T$ contains $O(n+k)$ triples---at most one associated with each edge of the trees, and at most 
\fchanged{two}
associated with each half-polygon $H \in {\cal H}$.

\begin{lemma}\label{lemma:coarse_cover}
$\cal T$ is a coarse cover.
\end{lemma}
\begin{proof}
There are three properties for a coarse cover.
For each triple $(I,f,H)$ it is clear that $I$ is a subinterval of $K$ and $f$ is defined on $I$---this is property (1).  For property (2), 
\newestchanged{$f$ is defined to have one of the three forms.}
\changed{Furthermore, we claim that for each triple 
$(I,f,H)$, and each $x \in I$, 
$f(x) = d(x,H)$.  
\fchanged{This is clear for Case 0.  For the other three Cases,} 
the triangle or trapezoid is part of the shortest path map $M(H)$, so the formula for $f(x)$ matches the distance $d(x,H)$.}

Finally, we must prove property (3).
\newestchanged{Let ${\cal T}^0$ be the initial large coarse cover defined above, consisting of the set of triples $(I,f,H)$ from Case 0 and
the union of all the
\newestchanged{triples arising from the}
shortest path maps $\{M(H): H \in {\cal H}\}$. 
Then ${\cal T} \subseteq {\cal T}^0$, and 
we must show that no triple of ${\cal T}^0$ that is omitted from $\cal T$ causes a violation of property (3).}
Any triple from the shortest path maps corresponds to a triangle that arises in Case 1 or 2 (i.e., with the same $u,v,I$), or to a trapezoid considered in Case 3. 
No trapezoids are omitted in $\cal T$, so it suffices to consider Cases 1 and 2.

We first examine Case 1.
Consider an edge $(u,v)$ in ${\bar T}_a$ where $u = p_a(v)$ and $u$ and $v$ are both visible from $K$, and consider the interval $I = [x_a(u),x_a(v)]$.
Suppose that for some $f,H$, there is a triple $(I,f,H)$ that is
included in the triples from the shortest path maps, but omitted from $\cal T$. Then there is a directed path in ${\bar T}_a$ from $v$ to a leaf corresponding to $H$, and $f(x)$ is $d_2(x,u) + |uv| + \kappa$ where $\kappa$ is  the length of the tree path from $v$ to the leaf corresponding to $H$. 
But then $f(x) \le d_2(x,u) + |uv| + \ell_a(v)$, since $\ell_a(v)$ is the maximum distance from $v$ to a leaf
\newestchanged{corresponding to farthest half-polygon $F_a(v)$.  If the inequality is strict, then $H$ is not a farthest half-polygon from any $x \in I$ so property (3) is satisfied without the
triple $(I,f,H)$.
And if equality holds, then property (3) allows us to omit the triple $(I,f,H)$
since the triple $(I,f,F_a(v))$ has the same $I$ and $f$.}
The case of triples omitted in Case 2 is similar. 
\end{proof}

\subsubsection{Finding the Relative Geodesic Center on a Chord}
\label{sec:find-rel-center}

The last step of the chord oracle is exactly the same as in Pollack et al.~\cite{pollack_sharir}. 
Given a chord $K$ and the 
\changed{coarse cover $\cal T$ from Section~\ref{section:functions_to_capture}---which provides a set of $O(n+k)$ functions whose upper envelope is the geodesic radius function on chord $K$---we want to find the relative center, $c_K$, that minimizes the geodesic radius function.}  
Pollack et al.~use a technique of Megiddo's to do this in  $O(n+k)$ time
\newestchanged{by recursively reducing to a smaller subinterval of $K$ while eliminating elements 
of the coarse cover whose functions are strictly dominated by others.}
In brief, the idea is to pair up the functions, 
define a set of at most 6 ``extended intersection points'' for each pair, 
and test medians of those points in order to \changed{restrict the search to a subinterval of $K$ and} eliminate a constant fraction of the functions. %
\newestchanged{
Testing median points is done via
the test from Section~\ref{sec:testing_center} 
of whether the relative center is left/right of a query point $x$ on $K$.
This test depends on having  
the first segments of the shortest paths  from $x$ to its farthest half-polygons.
Observe that 
the initial coarse cover from Section~\ref{section:functions_to_capture} captures these segments, and they are preserved throughout the recursion because only strictly dominated functions are eliminated.
}

We fill in a bit more detail.
In each round we have a 
subinterval $K'$ of $K$ that contains $c_K$ and a subset ${\cal T}'$ of the coarse cover $\cal T$
\fchanged{such that any function omitted from ${\cal T}'$ is strictly dominated on interval $K'$ by a function of ${\cal T}'$}.
\newestchanged{We want to eliminate a constant fraction of ${\cal T}'$ in time $O(|{\cal T}'|)$.}
We pair up the functions of ${\cal T}'$. 
Consider a pair of functions $f_1$ and $f_2$.  Each function is defined on a subinterval of $K$ and we define it to be $- \infty$ outside its interval. 
The upper envelope of $f_1$ and $f_2$ switches between $f_1$ and $f_2$ at \defn{extended intersection points} which include the points where $f_1$ and $f_2$
intersect \fchanged{(are equal)}, 
and also possibly the endpoints of their intervals.  
\newestchanged{If a subinterval of $K$ does not contain an extended intersection point for $f_1$ and $f_2$, then}
one of $f_1, f_2$ is irrelevant because it is dominated by the other (or both are $- \infty$). 
We know from Section~\ref{section:functions_to_capture} that each function 
has the form \fchanged{$f(x)  =0$} or $f(x) = d_2(x,s) + \kappa$ where $\kappa$ is a constant,
$d_2$ is Euclidean distance, and  
$s$ is 
a point or line. 
This implies that there are at most two intersection points of the functions $f_1$ and $f_2$, and thus at most six extended intersection points.  In fact, a closer examination shows that there are at most four extended intersection points.  

Pollack et al.~show how to successively test three medians of extended intersections in order to reduce the interval $K'$ and eliminate a constant fraction of the functions of ${\cal T}'$.
\fchanged{The first median test reduces the domain to a subinterval containing half the extended intersections, so three successive median tests reduce the domain to a subinterval containing one eighth of the extended intersections.  This implies that for at least half the pairs $f_1, f_2$, all four of their extended intersections lie outside the domain, and one of $f_1, f_2$ is dominated by the other and can be eliminated.}

This completes one round of their procedure, 
\newestchanged{with a runtime of $O(|{\cal T}'|)$.} 
When ${\cal T}'$ is reduced to constant size, the relative center $c_K$ can be found directly.
The total run time is then $O(n+k)$.

\subsection{Finding the Geodesic Center of Half-Polygons}\label{section:center_using_oracle}
In this section we show how to use the 
$O(n+k)$ time chord oracle from Section~\ref{section:chord_oracle} to find the geodesic center of the $k$ half-polygons in $O((n+k) \log (n+k))$ time.
The basic structure of the algorithm is the same as that of Pollack et al.~\cite{pollack_sharir}. 

In the first step we use the chord oracle to restrict 
the search for the geodesic center
to a small region where the problem reduces to a Euclidean problem 
\changed{of finding a minimum radius disk that intersects some half-planes and contains some disks}.
\changed{This step takes $O((n+k) \log (n+k))$ time.}
\newcomment{In the second step we solve the resulting Euclidean problem
\fchanged{in linear time}, which}
involves some 
new ingredients to handle our case of half-polygons.

\subsubsection{Finding a Region that Contains the Geodesic Center}
\label{section:restrict}
\label{section:triangle}

Triangulate $P$ in linear time~\cite{chazelle1991triangulating}.
Choose a chord of the triangulation that splits the polygon into two subpolygons so that the number of triangles on each side is balanced \changed{(the dual of a triangulation is a tree of maximum degree 3, which has a balanced cut vertex).} 
Run the chord oracle on this chord, and recurse in the appropriate subpolygon. 
In $O(\log n)$ iterations, we narrow our search down to one triangle $T^*$ of the triangulation. This step takes $O((n + k) \log n)$ time.

Next, we refine $T^*$ to a %
region
$R$ that contains the center and such that 
$R$ is \defn{homogeneous}, meaning that
for any $H \in \cal H$ the shortest paths from points in $R$ to $H$ have the same combinatorial structure, i.e., the same sequence of polygon vertices along the path.

The idea is to subdivide $T^*$ by $O(n+k)$ lines so that each cell in the resulting line arrangement is homogeneous, and then to find the cell containing the center.   
Construct the shortest path trees to $\cal H$ from each of the three corners of triangle $T^* = (a^*,b^*,c^*)$ using the algorithm of Section~\ref{section:shortest_path_tree_half_polygons}. For each edge $(u,v)$ of each tree, add the line through $uv$ if it intersects $T^*$.  (In fact, we do not need all these lines---\changed{as in the construction of the coarse cover in Section~\ref{section:functions_to_capture}, it suffices to use tree edges $(u,v)$ such that $u$ is visible from an edge of $T^*$.)}
We add three more lines for each  half-polygon $H \in {\cal H}$, specifically,  
the chord $h$ that defines $H$,
and the two lines perpendicular to $h$ through the endpoints of $h$.
The result is a set $L$ of $O(n+k)$ lines that we obtain in time $O(n+k)$. It is easy to 
see
that the resulting line arrangement has homogeneous regions.

All that remains is to find the cell of the arrangement that contains the geodesic center. 
It is simpler to state the algorithm in terms of $\epsilon$-nets 
instead
of the rather involved description of Megiddo's technique used by Pollack et al.~\cite{pollack_sharir}.
\changed{
For background on $\epsilon$-nets see the survey by 
Mustafa and Varadarajan~\cite[Chapter 47]{toth2017handbook} or the book by Mustafa~\cite{mustafa2022sampling}.}
\changed{The 
\newestchanged{high-level} idea is to define a range space with ground set $L$ and to find a constant-sized $\epsilon$-net in time $O(|L|)$.
Then the lines of the $\epsilon$-net divide our region into a %
\newestchanged{constant}
number of subregions and we can find which subregion contains the geodesic center by applying the chord oracle $O(1)$ times.  By the property of $\epsilon$-nets the subregion is intersected by only a constant fraction of the lines of $L$, 
so repeating this step for $O(\log (n+k))$ times, we arrive at a region $R$ with the required properties.}

\changed{We fill in a bit more detail.
The range space has ground set $L$. To define the ranges, let}
$\mathcal{T}$ be the 
\newestchanged{(infinite) set of all 
triangles contained in} $T^*$.
\revised{For $t \in \mathcal{T}$, let $\Delta_t = \{ \ell \in L : \ell \ {\rm intersects}\ t\}$.
Let $\Delta = \{ \Delta_t: t \in {\cal T} \}$.}
\changed{Then the range space is $S = (L,\Delta)$.
To show that constant-sized $\epsilon$-nets exist, we must show that $S$ has constant VC-dimension, or constant shattering dimension. 
We argue
that the shattering dimension is 6, i.e., that for any subset $L'$ of $L$ of size $m$ the number of ranges is $O(m^6)$. 
The lines intersecting a triangle $t$ are the same as the lines intersecting the convex hull of the three cells of the arrangement of $L'$ that contain the endpoints of $t$.  There are $O(m^2)$ cells in the arrangement and we choose three of them, giving the bound of $O(m^6)$ possible ranges.}
Thus, a constant sized $\epsilon$-net of size $O(\frac{1}{\epsilon} \log (\frac{1}{\epsilon}))$ exists for our range space.
\changed{In order to construct an $\epsilon$-net in deterministic $O(|L|)$ time, we need a subspace oracle that, given a subset $L'$ of $L$ of size $m$, computes the set of ranges of $L'$ in time proportional to the output size, 
$O(m^{6+1})$.  Begin by finding, for each line in $L'$, which cells of the arrangement lie to each side of the line.  Then, for every choice of three cells (there are $O(m^6)$ choices),  the lines intersecting their convex hull can be listed in time $O(m)$ time.}

\changed{For the algorithm, we choose $\epsilon = \frac{1}{2}$, and construct an $\epsilon$-net $N$.}
\changed{Triangulate each cell of the arrangement of $N$---this is a constant time operation since the arrangement has constant complexity.}
We can locate the triangle $T'$ of this triangulated arrangement that contains the visibility center in $O(n+k)$ time by running the chord oracle a constant number of times.
By the $\epsilon$-net property, no more than $\epsilon \cdot |L|$ lines of $L$ intersect the interior of $T'$.
Thus, in $O(|L|)$ time, we have halved the number of lines going through our domain of interest (the region that contains the geodesic center).
Repeating the same sequence of steps $O(\log |L|) = O(\log (n+k))$ times, we will arrive at a triangle containing a constant number of these lines.
At this point, a brute force method suffices to locate a
region $R$ that satisfies the properties stated in the beginning of this section.

In each iteration of the process, we %
apply the $O(n+k)$ chord oracle a constant number of times and thus the 
total runtime for this step is
$O((n+k) \log (n+k))$.

\subsubsection{Solving an Unconstrained Problem}\label{section:unconstrained_problem}

At this point, 
\changed{we have a %
homogeneous polygonal region $R$ that}
contains a geodesic center of the set $\mathcal{H}$ of half-polygons. 
Our goal is to find
the point $x \in R$ that minimizes the maximum over $H \in {\cal H}$ of $d(x,H)$. 
\newestchanged{We give a linear time algorithm (in this final step there is no need for an extra logarithmic factor).}
\changed{We show that the problem reduces to one in the Euclidean plane, i.e., the polygon no longer matters.}
Pick an \changed{arbitrary} point $p$ in $R$ and find the shortest path tree from $p$ to all half-polygons (this takes linear time). 
\changed{If $p$ has distance 0 to half-polygon $H$, then the same is true for all points in $R$, so $H$ is irrelevant and can be discarded. 
If $\pi(p,H)$ \newestchanged{consists of a single line segment that} reaches an internal point of the chord defining $H$ (we denote these half-polygons by ${\cal H}_1$), then 
$d(x,H) = d_2(x, {\bar H})$ for all $x \in R$, where $\bar H$ is the half-plane defined by $H$.  
And if the first segment of $\pi(x,H)$ %
reaches
a vertex $u$ (we denote these half-polygons by ${\cal H}_2$), then $d(x,H) = d_2(x,u) + \kappa$ for all $x \in R$, where $\kappa$ is a constant.
}
\changed{Thus we seek a point $x = (x_1,x_2)$
and a value $\rho$ 
to solve:
\begin{equation*}
\begin{array}{ll@{}ll}
\text{minimize}  & \rho &\\
\text{subject to}&
d_2(x,{\bar H}) \le \rho & H \in {\cal H}_1\\

&d_2(x,u) + \kappa \le \rho \ \ \ &
\text{ for point $u$ and constant $\kappa$ corresponding to } H \in {\cal H}_2\\
\end{array}
\end{equation*}
Because $x$ is guaranteed to lie in the region $R$, 
we can completely disregard the underlying polygon $P$ in solving the problem.
}

In the Euclidean plane, the problem may be reinterpreted in a geometric manner.
We wish to find the
\changed{disk}
of smallest radius $\rho$ that  \textit{intersects} %
each of a given set of \changed{half-planes} %
and \textit{contains} %
each of a given set of disks.
\changed{For $H \in {\cal H}_1$, we have $d(x,H) \le \rho$ if and only if the disk of radius $\rho$ centered at $x$ intersects $\bar H$.
For $H \in {\cal H}_2$, with $d(x,H) = d_2(x,u) + \kappa$, we have 
$d(x,H) \le \rho$ if and only if the disk of radius $\rho$ centered at $x$ contains the disk of radius $\kappa$ centered at $u$. 
}
\changed{We will call this Euclidean problem the ``minimum feasible disk'' problem.}
\newestchanged{The constraints of the problem that correspond to the set of half-polygons ${\cal H}_1$ will be referred to as \defn{half-plane constraints}, while the constraints for ${\cal H}_2$ will be called \defn{disk constraints}.}

\changed{We observe here that the minimum feasible disk problem belongs to the class of `LP-type' problems described by Sharir and Welzl~\cite{sharir1992combinatorial}.
In fact, it satisfies the computational assumptions that allow a derandomization of the Sharir-Welzl algorithm yielding a \textit{deterministic} linear-time algorithm for the problem (see Chazelle and Matousek~\cite{chazelle1996linear}).
However, as this approach is rather complex, we will outline a more direct linear-time algorithm to solve the problem.
 }

\changed{The minimum feasible disk problem is a combination of two well-known problems that have linear time algorithms. 
\newestchanged{If all the constraints are half-plane constraints,}
then, because each such constraint 
can be written as a linear inequality, 
we have a 3-dimensional linear program, which can be solved in linear time as shown by Meggido~\cite{megiddo_linear}
and independently by Dyer~\cite{dyer1984linear}.
On the other hand, 
\newestchanged{if all the constraints are disk constraints,}
then this is the ``spanning circle problem''---to find the smallest disk that contains some given disks.  This problem  arose from the geodesic vertex center problem~\cite{pollack_sharir} and generalizes the Euclidean 1-center problem where the disks degenerate to points.  The problem was solved in linear time by Megiddo~\cite{megiddo_spanned_ball} using an approach similar to that for the 1-center problem and for linear programming.
Because the approaches are similar, it is not difficult to combine them, as we show below.
}

\changed{We begin by describing the main idea of Megiddo's prune-and-search approach for both  linear programming
\newestchanged{in 3D}
and for the spanning circle problem (also see the survey by Dyer et al.~in the Handbook of Discrete and Computational Geometry~\cite[Chapter 49]{dyer2017linear}). %
The goal is to spend linear time to prune away a constant fraction of the constraints that do not define the final answer, and to repeat this until there are only a constant number of constraints left, after which a brute force method may be employed. 
The idea is to pair up the constraints, and for each pair of constraints $c_1,c_2$ compute a ``bisecting'' plane $\Pi$ such that on one side of the $\Pi$ the constraint $c_1$ is redundant, and on the other side of $\Pi$ the constraint $c_2$ is redundant. 
If we could identify which side of $\Pi$ contains the optimum solution, then one of the constraints $c_1, c_2$ can be removed. 
We address the
existence of such bisecting planes  below.   
There are two other issues.  Issue 1 is to identify which side of a plane contains the optimum solution, a subproblem that Megiddo calls an ``oracle''.  This is done by finding the optimum point restricted to the plane (a problem one dimension down), from which the side of the plane can be decided. (The Chord Oracle from Section~\ref{section:chord_oracle} was doing a similar thing.)  Issue 2 is to identify the position of the optimum point relative to ``many'' of the bisecting planes, while testing only a ``small'' sample of them---\newestchanged{this can be done using \emph{cuttings}.}  
We will not discuss these two issues since they are the same as in Megiddo's papers~\cite{megiddo_linear,megiddo1984linear} (or see the survey by Dyer et al.~\cite{dyer2017linear}).}

\changed{For our minimum feasible disk problem, %
\newestchanged{we have two types of constraints---half-plane constraints and disk constraints.}
\newestchanged{Megiddo's prune-and-search approach based on pairing up the constraints can still be applied so long as we pair each constraint with another constraint of the same type.}
\newestchanged{(We note that this idea was previously used by Bhattacharya et al.~\cite{bhattacharya1994optimal} in their linear time algorithm to find the smallest disk that contains some given points and intersects some lines, a problem they call the ``intersection radius problem''.)}
Thus it suffices to describe what are the bisecting planes for the two types of constraints in our minimum feasible disk problem.}

\changed{
\newestchanged{A half-plane constraint}
has the form 
$d_2(x, {\bar H}) \le \rho$. If the 
halfplane $\bar H$ is given by $a_i^T x  \le b_i $,
normalized so that $||a_i|| = 1$, 
then the constraint is $a_i^T x \le b_i + \rho$, a linear inequality. 
For two such constraints indexed by $i$ and $j$, the bisecting plane is given by $(a_i^T - a_j^T) x - (b_i - b_j) = 0$.
} 

\changed{
\newestchanged{A disk constraint}
has the form $d_2(x,u_i) + \kappa_i \le \rho$, corresponding to a disk with center $u_i$ and radius $\kappa_i$.  
As Megiddo~\cite{megiddo_spanned_ball} noted, by adding the constraint $\rho \ge \kappa_i$,  this can be written as 
$$ ||x - u_i||^2 \le (\rho - \kappa_i)^2
$$
or as 
$$f_i(x,\rho) \le 0$$ where $f_i$ is defined as 
$$ f_i(x,\rho) = ||x||^2 - 2u_i^T x + ||u_i||^2 - \rho^2 + 2\kappa_i \rho - \kappa_i^2.$$
This is not a linear constraint, but
for $i \ne j$, the  equation $f_i(x,\rho) = f_j(x,\rho)$ defines a plane since the quadratic terms, $||x||^2$ and $\rho^2$, cancel out.  So the bisecting plane is 
$f_i(x,\rho) = f_j(x,\rho)$.}

\newestchanged{This completes the summary of how to solve the minimum feasible disk problem in linear time, and completes our algorithm to find the geodesic center of half-polygons.}

\bigskip

\section{Conclusions}

We introduced the notion of the visibility center of a set of points in a polygon and gave an algorithm with run time $O((n+m) \log (n+m))$ to find the visibility center of $m$ points in an $n$-vertex polygon.  To do this,  we gave an algorithm \changed{with run time $O((n+k)\log(n+k))$} to find the geodesic center of a given set of $k$ half-polygons inside a  polygon, a problem of independent interest. We conclude with some open questions.

Can the visibility center of a simple polygon be found more efficiently? 
Note that the geodesic center of the vertices of a simple polygon can be found in linear time~\cite{linear_time_geodesic}. 
Our current method involves ray shooting
and sorting (Section~\ref{section:essential-half-polygons} and the preprocessing in Section~\ref{section:half-polygon-center}) %
, which are serious barriers.  A more reasonable goal
\newcomment{is to find the visibility center of $m$ points in a polygon in time $O(n + m \log m)$.}

Is there a more efficient algorithm to find the geodesic center of (sorted) half-polygons?
In forthcoming work we give a linear time algorithm for the special case of finding the geodesic center of the \emph{edges} of a polygon (this is the case where the half-polygons hug the edges).

How hard is it to find 
the farthest visibility Voronoi diagram of a polygon? 
Finally, what about the 2-visibility center of a polygon, where we can deploy two guards instead of one?

\bibliographystyle{plainurl}   
\bibliography{RS_refs_1}

\newpage

\begin{appendix}

\end{appendix}

\end{document}